\documentclass[11pt]{article}
\usepackage[utf8]{inputenc}
\usepackage{amsthm}
\usepackage{mathtools}
\usepackage{graphicx} 
\usepackage{array} 
\usepackage{tabularx, booktabs}
\newcommand{\midsepremove}{\aboverulesep = 0mm \belowrulesep = 0mm}
\midsepremove

\usepackage[top=1.in, bottom=1.in, left=1.in, right=1.in]{geometry}

\usepackage{amsmath, amssymb, amsfonts, verbatim}
\usepackage{hyphenat,epsfig,subfigure,multirow}
\usepackage{nicefrac}
\usepackage[usenames,dvipsnames]{xcolor}
\usepackage[ruled, vlined, noend]{algorithm2e}
\SetKw{Continue}{continue}

\usepackage{mdframed}
\usepackage{enumitem}

\definecolor{darkpastelgreen}{rgb}{0.01, 0.75, 0.24}
\usepackage[hidelinks]{hyperref}     
\usepackage[noabbrev,capitalise]{cleveref}
\hypersetup{colorlinks={true},linkcolor={blue},citecolor=darkpastelgreen}

\usepackage{tcolorbox}
\usepackage{url}
\usepackage{thmtools}

\newtheorem{theorem}{Theorem}
\newtheorem{lemma}{Lemma}[section]

\newtheorem{fact}[lemma]{Fact}

\theoremstyle{definition}
\newtheorem{definition}[lemma]{Definition}

\newtheorem*{claim*}{Claim}
\newtheorem*{proposition*}{Proposition}
\newtheorem*{lemma*}{Lemma}
\newtheorem*{corollary*}{Corollary}
\newtheorem*{remark*}{Remark}

\def\vq{{\mathbf{q}}}
\def\vp{{\mathbf{p}}}
\def\vs{{\mathbf{s}}}
\def\vz{{\mathbf{z}}}
\def\eps{\varepsilon}
\def\epsilon{\varepsilon}
\def\cT{{\mathcal{T}}}

\newcommand{\wrap}[1]{\parbox{1.017\linewidth}{\centering\vspace{1.5mm}#1\vspace{1mm}}}

\newcommand\fang[1]{\textcolor{teal}{[F: #1]}}

\DeclareMathOperator{\poly}{poly}
\DeclareMathOperator*{\argmax}{argmax}
\allowdisplaybreaks

\SetKwInput{KwData}{Input}
\SetKwInput{KwResult}{Output}

\title{Hardness and Approximation Algorithms for\\ Balanced Districting Problems}
\author{Prathamesh Dharangutte\thanks{Department of Computer Science,
Rutgers University. Email: \{prathamesh.d,jg1555\}@rutgers.edu. Dharangutte and Gao are supported by NSF IIS-2229876, DMS-2220271, DMS-2311064, CCF-2208663,  CCF-2118953.} \and 
Jie Gao$^\ast$ \and
Shang-En Huang\thanks{
National Taiwan University. Email: math.tmt514@gmail.com. Shang-En Huang conducted part of this research in Boston College and in visiting Osaka University. Shang-En is currently supported by NTU Grant 113L7491.} \and
Fang-Yi Yu\thanks{
George Mason University, Email: fangyiyu@gmu.edu
}}
\date{}

\begin{document}

\maketitle

\begin{abstract}
We introduce and study the problem of balanced districting, where given an undirected graph with vertices carrying two types of weights (different population, resource types, etc) the goal is to maximize the total weights covered in vertex disjoint districts such that each district is a star or (in general) a connected induced subgraph with the two weights to be balanced. This problem is strongly motivated by political redistricting, where contiguity, population balance, and compactness are essential. We provide hardness and approximation algorithms for this problem. In particular, we show \textsf{NP}-hardness for an approximation better than $n^{1/2-\delta}$ for any constant $\delta>0$ in general graphs even when the districts are star graphs, as well as \textsf{NP}-hardness on complete graphs, tree graphs, planar graphs and other restricted settings. On the other hand, we develop an algorithm for balanced star districting that gives an $O(\sqrt{n})$-approximation on any graph (which is basically tight considering matching hardness of approximation results), an $O(\log n)$ approximation on planar graphs with extensions to minor-free graphs. Our algorithm uses a modified Whack-a-Mole algorithm [Bhattacharya, Kiss, and Saranurak, SODA 2023] to find a sparse solution of a fractional packing linear program (despite exponentially many variables) which requires a new design of a separation oracle specific for our balanced districting problem.
To turn the fractional solution to a feasible integer solution, we adopt the randomized rounding algorithm by [Chan and Har-Peled, SoCG 2009]. To get a good approximation ratio of the rounding procedure, a crucial element in the analysis is the \emph{balanced scattering separators} for planar graphs and minor-free graphs --- separators that can be partitioned into a small number of $k$-hop independent sets for some constant $k$ --- which may find independent interest in solving other packing style problems. Further, our algorithm is versatile --- \emph{the very same algorithm} can be analyzed in different ways on various graph classes, which  leads to class-dependent approximation ratios. We also provide a FPTAS algorithm for complete graphs and tree graphs, as well as greedy algorithms and approximation ratios when the district cardinality is bounded, the graph has bounded degree or the weights are binary. 
\end{abstract}

\clearpage
\newpage


\setcounter{page}{1}

\section{Introduction}

In this paper we study the problem of balanced districting, where we are given an undirected graph where vertices carry two types of weights (population, resource types, etc) and the goal is to find vertex disjoint districts such that each district is a connected induced subgraph with the total weights to be $c$-balanced -- each type of weight is at least $1/c$ of total weight of the district. We aim to maximize the total weights of the vertex disjoint balanced districts.

This problem is an abstraction of many real world scenarios of districting where contiguity/connectivity, population/resource type balance and compactness are desirable properties. 
For example, in political redistricting, several towns are grouped into a state legislative district or a congressional district. 
Balancedness requires that each district maintains a sufficient fraction of each (political or demographic) group, which is essential for several reasons. First, voter turnout rates sharply increase if the anticipated election outcome is expected to be a close tie.~\cite{Bursztyn2023-dg} Thus a balanced district would motivate and raise voter turnout rates. Additionally, balancedness ensures that each political group has the opportunity to elect a candidate of their choice, in compliance with the Voting Rights Act of 1965~\cite{hebert2010realist} and other amendments~\cite{clelland2022colorado}. This principle also helps to prevent the tipping point in racial segregation, where residents of one demographic group start to leave a district once their population falls below a certain threshold.~\cite{schelling1971dynamic,orfield2020black}
Connectivity or contiguity, on the other hand, demands that each district be geographically contiguous,  -- in the language of graph theory that the vertices corresponding to the towns form an induced connected subgraph. This requirement is enforced by most state laws and is a standard practice in general. Many states also have a compactness rule~\cite{altman1998traditional}, which refers to the principle that the constituents residing within an electoral district should live as near to one another as possible. It often manifests into a preference for regular geometric shapes or high roundness (small ratio of circumference and total area).

The problem of redistricting also appears in many other scenarios such as districting for public schools, sales and services, healthcare, police and emergency services~\cite{bucarey2015shape}, and logistics operations~\cite{kalcsics2005towards}. 
In addition, the balanced districting problem is of interest in a broader range of applications for resource allocation. For example modern computing infrastructure such as cloud computing provides services to a dynamic set of customers with diverse demands. Customer applications may have a variety of requirements on the combination of different resources (such as CPU cycles, memory, storage, or access of special hardware) that can be summarized by the balanced requirement.

Due to the importance of the problem, redistricting has been studied in a computational sense for schools and elections, that dates back to the 1960s.~\cite{hess1965nonpartisan} Since then, an extensive line of work (see \cite{becker2020redistricting} for a survey) has formulated the redistricting task as an optimization problem with a certain set of objectives. A lot of existing work considers the geographical map as input and comes up with practical methods and software implementations that generate feasible districting plans. We will survey such work in \Cref{subsec:related}. 
However, most previous works focus on optimizing a single desirable property alone (e.g.,  connectivity~\cite{10.1145/3490486.3538271}, or balance~\cite{gillani2023redrawing}), or optimizing average aggregated scores combining multiple objectives of the districts~\cite{deford2021recombination}.  In contrast, our problem formulation takes these objectives as hard constraints and optimizes the total population that satisfies them.  There are several merits of this formulation.
First, it offers interpretable, fair, worst-case guarantees for the identified districts. Districting problem is a multi-faceted one. With multiple criteria taken into consideration, it feels ill-fit if only one criterion is singled out as the optimization objective. Furthermore, an average quality guarantee does not provide meaningful utility at the individual district  level, and aggregated scores offer limited insight into each objective.   
Second, a districting solution has a consequential nature and should be taken with a dynamic and time-evolving perspective. Once a new districting plan is in place, residents naturally respond to the algorithm output, resulting in changes in the population distributions. One prominent example is the tipping point theory in racial segregation mentioned above.  Optimizing a single balancedness score can still lead to many districts falling below such a tipping point, exacerbating segregation. With this in consideration it is important to keep balancedness as a hard constraint, which hopefully facilitates district stability and integration.
With $c$-balanced property as a requirement, a graph partitioning into vertex disjoint $c$-balanced districts is not always possible. For example, if the total weight is not $c$-balanced, some districts have to be unbalanced no matter how the districts are defined. Therefore, we aim to maximize the total weights in balanced districts. 


In this paper we focus on the graph theoretic perspective of the redistricting problem. We abstract the input as a graph where vertices represent natural geographical entities/blocks (e.g., townships) and edges of two vertices represent geographical adjacency/contiguity. We focus on two important quality considerations namely connectivity and balanceness, and we maximize coverage, i.e., the  total weight (population) covered by balanced districts. In addition, we also consider compactness, which in our setting leads to preference of districts as low-diameter subgraphs. An important case studied in this paper is to consider a balanced \emph{star district}, which consists of a center vertex $v$ as well as a set of neighboring blocks all adjacent to $v$. We also consider districts of bounded rank $k$ for a constant $k$ -- where a district has at most $k$ vertices.

\subsection{Our Results and Technical Overview}

We report a systematic study of the balanced districting problem on both hardness results and approximation algorithms. Our goal is to dissect the problem along different types of input graph topologies (general graphs, planar graphs, bounded degree graphs, complete graphs, tree graphs, etc), district types (e.g., arbitrarily connected districts, star districts, or bounded rank-$k$), and weight assumptions (arbitrary weights, binary weights).  A brief summary of our results can be found in~\Cref{tab:results}.


\begin{table}[htbp]
\caption{A summary of hardness and approximation results on the balanced districting problem. Note that $\delta, \epsilon > 0$ are constants. Tight results are highlighted in bold.}\label{tab:results}
    \centering
    \small
    \begin{tabularx}{\textwidth}{| c | >{\centering\arraybackslash}p{3cm} | >{\centering\arraybackslash}p{2.8cm} | >{\centering}X | l |}
\toprule
& Graph Type & District Type &  Result  & Note \\
\toprule
\multirow{4}{*}{\Cref{sec:hardness}} & General & arbitrary/star & \textbf{$\mathsf{NP}$-hard for $n^{1/2-\delta}$-approx} & \Cref{thm:apx}\\ \cline{2-5} 
 & \wrap{Max degree $\Delta$} &   arbitrary/star &  \wrap{$\mathsf{APX}$-hard for $\Delta=O(1)$\\          $\mathsf{NP}$-hard for $\Delta / 2^{O(\sqrt{\Delta})}$-approx \\ $\mathsf{UGC}$-hard for $O(\Delta/ \log^2 \Delta)$-approx} & \Cref{thm:apx}\\ \cline{2-5}
 & Planar with $\Delta=3$
 & star with rank-$3$ & $\mathsf{NP}$-hard & \Cref{thm:planar_hard}\\ \cline{2-5} 
 & Complete or Tree & arbitrary/star & \textbf{$\mathsf{NP}$-hard} & \Cref{thm:complete}\\ \cline{2-5} 
\hline 
\multirow{4}{*}{\Cref{sec:LP-star}} & Planar   & star & $O(\log n)$-approx & \multirow{4}{*}{\Cref{thm:star-district}}
\\
\cline{2-4} 
 & $H$-Minor-Free   & star & $O(h^2\log n)$-approx, $h=|H|$ & 
\\
\cline{2-4} 
  & Outerplanar  & star & $O(1)$-approx & 
\\
\cline{2-4} 
 & General& star & \textbf{$\Theta(\sqrt{n})$-approx} & 
\\ 
\hline 
\multirow{2}{*}{\Cref{sec:FPTAS}} & Complete graph &  arbitrary/star   &  \textbf{FPTAS $(1+\epsilon)$-approx}  & \Cref{thm:fptas-complete}\\
\cline{2-5}  
 & Tree & arbitrary/star & \textbf{FPTAS $(1+\epsilon)$-approx} & \Cref{thm:tree}\\
\hline 

\multirow{4}{*}{\Cref{sec:variants}} & General  & rank-$2$ & polynomially solvable & \Cref{thm:bounded_k_2}\\
\cline{2-5} 
& General  & rank-$k$, $k>2$ & $k$-approx & \Cref{thm:general_k}\\
\cline{2-5} 
& Bounded degree $\Delta$ & star & $(\Delta+\frac{1}{\Delta})$-approx & \Cref{thm:bounded_degree} \\
\cline{2-5} 
& General graph with binary weights & \multirow{2}{=}{\centering star} & \multirow{2}{=}{\centering $c$-approx} & \multirow{2}{*}{\centering \Cref{lem:local_search}}\\

\bottomrule
\end{tabularx}
\end{table}

\paragraph{Complexity and Challenges.}
There are three elements in the balanced districting problem that make it challenging and interesting, from a technical perspective: 1) connectivity -- the induced subgraph of a district is connected 2) packing and coverage maximization -- no vertex belongs to two districts and we maximize the total weights of included vertices; 3) balancedness -- the two types of weights in a district need to be roughly balanced. 
These elements are shared with a number of well known hard problems, suggesting that our problem is also  computationally challenging. For example, the exact set cover problem asks if there is a perfect coverage and packing in a set cover instance. The packing element is shared with maximum independent set problem -- a vertex included will forbid its neighbors to be included. And the balancedness is shared with subset sum problem. Therefore by using the hardness of these problems we can show hardness and hardness of approximation of the balanced districting problem for a variety of graph classes. 
The hardness of the balanced districting problem is immediately shown by a reduction from exact set cover problem. 
By a reduction from maximum independent set problem, we can show that the balanced districting problem does not have an approximation of $n^{1/2-\delta}$ for any constant $\delta>0$ in a general graph of $n$ vertices unless $\mathsf{P}=\mathsf{NP}$ 
and is $\mathsf{APX}$-hard for bounded degree graphs.
From a reduction from the planar 1-in-3SAT problem, the balanced districting problem is $\mathsf{NP}$-hard for a planar graph, and even if each district has at most three vertices -- a crucial condition since the problem can be solved exactly by maximum weighted matching if each district is only allowed to have two vertices. 
If the input graph is a tree or a complete graph -- extremely simple topologies for which many $\mathsf{NP}$-hard problems can be solved in polynomial time, the balanced districting problem is still $\mathsf{NP}$-hard due to the balancedness requirement by a reduction from subset sum. All of the hardness proofs hold even if we limit the output districts to be only of a star topology. This set of hardness results can be seen from the Top section of~\Cref{tab:results}.


\paragraph{Greedy Methods on Special Cases.} On the positive side, it is natural to ask if existing techniques for solving or approximating these related problems can be borrowed for the balanced districting problem. The answer turns out to be often ``not really'', even if we only look for balanced star districts. The additional requirements in our problem often break some crucial steps. For example, the problem of maximum independent set has an easy $\Omega(n)$ lower bound if the input graph is sparse (or planar) --- thus a simple random greedy algorithm with conflict checking gives an easy constant approximation algorithm. 
But such lower bound no longer holds true for our problem if the weights are not balanced. Even for star districts, when the maximum degree is not bounded by a constant, the number of potential balanced districts can be exponential in $n$, the size of the network.

We show in~\Cref{sec:variants} that ideas using a greedy approach and local search method give us approximation algorithms, but only for very special cases. 
Namely, if the districts have rank-$k$, we can try the greedy maximum hypergraph matching to have a $k$-approximation to the optimal solution. 
When all weights are binary ($1$ or $0$), we can use a greedy algorithm with local search to get a $c$ approximation to the optimal $c$-balanced star districting solution. Similarly, if the graph has maximum degree $\Delta$, we can get a $(\Delta+\frac{1}{\Delta})$-approximation for $c$-balanced star districting.

\paragraph{LP Framework and Rounding.}
To really tackle the problem with arbitrary weights, for districts that are not limited by rank and graphs beyond constant bounded degree, we
first examine what we can do with complete graphs or tree graphs -- here the topology is made some of the simplest possible, and we would like to address the challenge from packing and balancedness. In this setting we can obtain FPTAS for both complete graphs and tree graphs (\Cref{sec:FPTAS})-- although the algorithm is much more involved due to the additional requirements of packing and connectivity (for tree graphs, as connectivity is trivial for complete graphs). 
Our FPTAS uses a dynamic programming technique that maintains one district's possible weights and introduces a new prioritized trimming method to approximate weights while ensuring that the resulting district satisfies the $c$-balanced constraint and approximates the optimal weight. The FPTAS for the complete graph is later used as a subroutine for solving the relaxed LP formulation for other graph settings. 

Beyond complete graphs and tree graphs, we develop a general framework (\Cref{sec:LP-star}) that produce approximation algorithms for star districts on different types of graphs. 
All these algorithms start from a relaxed linear program where we formulate a variable $x_S$ for each potential balanced star district $S$, which can take non-integer values and for all districts that share the same vertex, the sum of their variables is at most $1$. 
Despite potentially exponentially many variables (and constraints in the dual program) that preclude standard solutions, we adapt the whack-a-mole framework~\cite{BKS23-lp}, which can be seen as a lazy multiplicative weight update algorithm, on dual variables, and we design a `separation oracle' that selects a violating constraint in the dual program in time polynomial in $n$ and $1/\eps$ that can significantly improve the solution.  Consequentially, the linear program can be solved in time polynomial in $n$ and $1/\eps$ up to any precision $1-\eps$ and the number of non-zero primal variables (i.e., the candidate balanced star districts) is also polynomial.  Intriguingly, our separation oracle is based on our FPTAS on compute graphs for balanced districting.

Now we will round the fractional solution to an integer solution and in the process we may lose an approximation factor. We use a simple randomized rounding method where we sort the districts with non-zero values in decreasing order of total weight, and flip a coin with probability proportional to $x_S$ to include a potential district $S$, if $S$ does not overlap with any districts already included. In order to bound the loss of quality in the rounding process, we need to upper bound the correlation of the variables, namely, sum of $x_A \cdot x_B$ for all pairs of overlapping districts $A, B$. These are the (fractional) districts that have to be dropped due to conflict. To establish the approximation factor, we wish to bound the total sum of correlation by a factor multiplied with the total sum over all possible districts $\sum_S x_S$ -- exactly the optimal LP solution. 
We show that this ratio is $O(\sqrt{n})$, which immediately gives an $O(\sqrt{n})$-approximate solution for star districts on a general graph. Notice that this is tight due to the hardness of approximation results.

Due to the strong motivation from political redistricting and resource allocation considering geographical proximity/constraints, the planar graph is of particular interest to us. One of the main technical contributions is a polylogarithmic-approximation algorithm for balanced star districting on planar graphs and related algorithms for minor-free graphs and outer planar graphs. 
For a planar graph we now adopt a balanced planar separator and use a divide-and-conquer analysis. Namely, we only need to analyze the overlapping districts with at least one of them including vertices in the separator. Now a crucial observation is, if we can partition the planar separator into $k$ 5-hop independent sets, then we can decompose the total sum of correlation by the independent sets -- fix an 5-hop independent set $X$, two star districts that touch different vertices in $X$ are disjoint and star districts that share the same vertex in $X$ have their total district value bounded by $1$ due to the primal constraint. This allows us to upper bound the correlation term for the star districts touching the separator by a factor of $k$ of the sum of $x_S$ with districts $S$ on the separator. Recursively, this gives an $O(\log n)$ factor loss in the final approximation. 

We remark that the above analysis asks for a new property of a balanced separator -- one that can be decomposed into a small number of 5-hop independent sets (called a ``scattering'' separator) -- and we do not care about the size of the separator. This is possible for a planar graph if we use the fundamental cycle separator, which is composed of two shortest paths, and thus at most $10$ $5$-hop independent sets. For a $H$-minor-free graph with $H$ as a graph of $h$ vertices, we show the existence of a similar separator, which can be decomposed into $O(h^2)$ 5-hop independent sets. Thus the final approximation ratio for $H$-minor-free graphs is $O(h^2 \log n)$. For outer planar graphs, we can skip the recursive step and work with graph partitions with $5$-hop independent sets and get $O(1)$-approximation. We believe that this technique of using balanced scattering separators is interesting in its own and may find additional applications in other problems with some packing (non-overlapping) requirement on the solution.

On general graphs, the formulated linear program could have an integrality gap as large as $\Omega(\sqrt{n})$.
Since our rounding algorithm turns an fractional solution into an integral one, this barrier unavoidably blends into our analysis to the proposed rounding algorithm, producing an provable $O(\sqrt{n})$ bound.
However, by thinking about this argument contrapositively, an upper bound to our rounding algorithm leads to the integrality gap of the formulated LP, which could be an interesting takeaway. On the other hand, we also show that there are planar graphs (specifically grid graphs) such that our rounding algorithm produces a $>1$ constant approximation ratio.
This observation suggests that we cannot hope for a PTAS using this approach, even on planar graphs.

\subsection{Related work}\label{subsec:related}
To the best of our knowledge, this paper is the first to study the balanced districting problem. Below, we briefly survey related problems and explore their potential connections to ours.

\paragraph{Districting.}
Our problem is connected to computational (re)districting for schools and elections, which dates back to the 1960s.~\cite{hess1965nonpartisan}  Since then, an extensive line of work (see \cite{becker2020redistricting} for a survey) has formulated the redistricting task as an optimization problem with a certain objective and constraints, e.g,. balancedness, contingency, or compactness.  Our redistricting problem focuses on optimizing the population in balanced and contiguous districts.
One concept related to our notion of balance is competitiveness.  Recent work introduces vote-band metrics~\cite{deford2020computational}, which require a certain fraction of votes to fall within a specified range (e.g., 45-55\%) for competitive elections. Subsequently, \cite{chuang2024drawing} also adopt similar notions called $\delta$-Vote-Band Competitive which is equivalent to our $c$-balancedness by setting $c = 2/(1-2\delta)$.  While related, our work diverges technically, offering both hardness and algorithmic results for several common graphs.  \cite{deford2020computational} empirically evaluates ensemble methods for district distributions.  \cite{chuang2024drawing} explored the hardness and heuristic algorithms for maximizing the number of districts meeting the target competitiveness constraints, with additional requirements that all districts have roughly the same population limited compactness consideration.  

One approach treats contiguity as a transportation cost and designs linear programming models to minimize this total cost~\cite{10.1145/3490486.3538271,franklin1973computed,clarke1968operations}.  Interestingly, the fair clustering problem can also be viewed as optimizing contingency~\cite{Bohm2020-kd, Chierichetti2018-qs, Jia2020-xu,Chhabra2021-nb}.   Other research focuses on optimizing compactness scores~\cite{bar2020gerrymandering,jin2017spatial,kim2011optimization,jacobs2018partial} or using Voronoi or power diagrams with some variant of $k$-means~\cite{weaver1963procedure,cohen2017balanced,cohen2018balanced,fryer2011measuring}.  Finally, another line of work optimizes balance scores \cite{gillani2023redrawing}.  These approaches differ from ours in that they treat specific aspects of districting (contiguity, compactness, balance, etc.) as objectives, rather than maximizing the population that meets these criteria.

Besides the optimization approach, another popular approach uses sampling to generate a distribution over districts and create a collection of district plans for selection.~\cite{altman2011bard,cirincione2000assessing}  One widely used method is ReCom~\cite{deford2021recombination}, an MCMC algorithm.  However, these approaches may suffer from slow mixing times and lack formal guarantees.~\cite{najt2019complexity}
Finally, several papers take a fair division approach~\cite{landau2009fair,pegden2017partisan,de2018analysis}.  The problem is quite different, however, as fairness is defined concerning parties (types) and the number of seats they would win (i.e., the number of districts where they would have a majority) compared to other districts.

\paragraph{Algorithms.}
Beyond districting problems, as outlined in the technical overview, our problem connects to several classical algorithm problems.  If we only want to maximize the population of a single connected and balanced district, the problem becomes a balanced connected subgraph problem~\cite{Bhore2022-fy,bhore2022balanced,martinod2021complexity}.  However, the previous work in this area typically considers unit weights for either type, which does not adequately represent the districting problem that operates in an aggregated block-level setting.
Our problem can be seen as packing subgraphs on graph~\cite{cornuejols1982packing}, e.g., edges (maximum matching), triangles~\cite{keevash2004packing}, circles~\cite{ding2002packing}. 
Finally, we note a line of work on a balanced, connected graph partition~\cite{chen2021approximation,chlebikova1996approximating}, and balanced bin-packing problem~\cite{falkenauer1992genetic}, which, however, aim to generate a partition where each component has similar weights.






\subsection{Open Problems}

As the first work to formally study the balanced districting problem in this formulation, our work leaves a number of interesting open problems for future work. Obviously it is good to close the gap of approximation and hardness for different families of graphs. Our results are tight for general connected districts on complete graphs and tree graphs, as well as star districts on general graphs, but leave gaps for other settings.
We conjecture that there exists an algorithm with a constant approximation factor for $c$-balanced star districting on planar graphs. 
It would also be interesting to develop algorithms to go beyond star districts, i.e., $k$-hop graphs for a constant $k$ or the more general setting of connected districts. We remark that the scattering separator can be modified to handle $k$-hop graphs but we need an efficient separation oracle. 
We consider two types of weights/populations and generalizing the problem and solutions to three or more colors would be interesting and is currently widely open. We remark that the PTAS algorithm for complete graphs is specific for two weights/colors.
Finally, an interesting future direction would be to develop algorithms that also demand  approximate population equality among districts.


\section{Preliminaries}

Let $G=(V, E)$ be an undirected graph where we call the vertices \emph{blocks}.
A \emph{district} $T\subseteq V$ is a subset of blocks where the induced subgraph $G[T]$ is connected. If there exists a block $x\in T$ that is a neighbor of every other block in $T\setminus\{x\}$, then we say $T$ is a \emph{star district} and $x$ is a \emph{center} of $T$.
The \emph{rank} of a district $T$ is the number of blocks in $T$.
A \emph{(partial) districting} $\mathcal{T}$ is a collection of disjoint districts. That is, $\mathcal{T} = \{T_1, T_2, \ldots, T_m\}$ where $T_i\cap T_j=\emptyset$ whenever $i\neq j$. Notice that a districting is not necessarily a partitioning of the graph, i.e., not all blocks are included in the districts.

\paragraph{$c$-Balanced Districting Problems.}
There are two communities or commodities of interest.
Let the functions $p_1, p_2:V\to \mathbb{Z}_{\ge 0}$ represent the \emph{population of each community} or \emph{weight of each commodity} on each vertex. 
The \emph{weight} of a block $w(x)$ is defined to be $p_1(x)+p_2(x)$.
Let $T\subseteq V$ be a district.
By a natural extension we define $p_i(T) := \sum_{x\in T} p_i(x)$ for $i\in\{1, 2\}$ and $w(T) := p_1(T)+p_2(T)$ accordingly as the weight of the district $T$.  
Finally, given a districting $\mathcal{T}$, we define $w(\mathcal{T}) := \sum_{T\in \mathcal{T}} w(T)$ to be the total weight. 

Given a constant $c\ge 2$, we say that a district $T$ is \emph{$c$-balanced} if
\begin{equation}\label{eq:c-balanced}
\min\{p_1(T), p_2(T)\}\ge \frac{w(T)}{c}\ .
\end{equation}
$\mathcal{T}$ is a \emph{$c$-balanced districting} if all districts $T\in \mathcal{T}$ are $c$-balanced. Notice that if the total weights in the graph are not $c$-balanced, we cannot hope to include all blocks in $c$-balanced districts. 
Given $c\ge 2$, a graph $G$, and functions $p_1$ and $p_2$,
the problem \textsc{$c$-Balanced-Districting} is subjected to find any $c$-balanced districting $\mathcal{T}$ that \emph{maximizes} $w(\mathcal{T})$. That is, we wish to maximize population covered in the $c$-balanced districts.

We will investigate several variants to the problem.
A lot of our results concern restricting the output districting to be \emph{star shaped}, respecting the need for compactness of the districts. We also consider districts of bounded rank $k$ if every district has at most $k$ vertices with $k$ assumed to be a constant. In general we consider the weights of the vertices to be arbitrary integer values. A special case is when all weights are uniform (binary) -- each vertex has only one non-zero weight type, either $p_1(x) = 0$, $p_2(x) = 1$ or $p_1(x) = 1$, $p_2(x) = 0$. 

Let $X$ be any variant to the $c$-balanced districting problem.
A districting $\mathcal{T}$ is said to be \emph{feasible} on $X$ if $\mathcal{T}$ satisfies all districting type constraints, but not necessarily to have its weight maximized. Any districting with the maximum possible total weight is said to be \emph{optimal}.
We say that a feasible districting $\mathcal{T}$ is an \emph{$f$-approximated} solution if $f\cdot w(\mathcal{T})\ge w(\mathcal{T}_{\textsf{OPT}})$, where $\mathcal{T}_{\textsf{OPT}}$ is any optimal districting.

\paragraph{Graph Types.} A graph $G=(V, E)$ is said to be \emph{planar} if there exists an embedding of all vertices to the Euclidean plane such that all edges can be drawn without intersections other than the endpoints.
A \emph{face} of an planar embedding is a connected region separated by the embedded edges. $G$ is said to be \emph{outerplanar} if there exists an embedding of $G$ such that there is a face containing all vertices. Often this face is assumed to be the outer face. A graph $H$ is said to be a \emph{minor} of $G$ if $H$ is isomorphic to the graph obtained by a sequence of vertex deletions, edge contractions, and edge deletions from $G$. We say that $G$ is \emph{$H$-minor-free} if $G$ does not have $H$ as its minor. 

\section{Hardness Results}\label{sec:hardness}

We first present hardness results for a variety of $c$-balanced districting problems with increasing restrictions on the parameters.  The proofs are deferred to \Cref{omit3}.

\begin{restatable}{theorem}{BasicHardnessResult}
\label{thm:hardness-basic}
The $c$-balanced districting problem is $\mathsf{NP}$-hard, for both the case when the the districts are connected subgraphs and when the districts are required to be stars.
\end{restatable}


\Cref{thm:hardness-basic} uses a reduction from the \textsc{ExactSetCover} problem.
The \textsc{ExactSetCover} problem remains $\mathsf{NP}$-hard even when each set has \emph{exactly} three elements and no element appears in more than three sets~\cite{Gary1979-pp}, or when each element appears in exactly three sets~\cite{Gonzalez1985-qk}. 
Therefore, if we limit that each district has at most \emph{four} blocks, the problem remains $\mathsf{NP}$-hard. In the following we show that the problem remains hard if each district has at most \emph{three} blocks and the graph is planar. Notice that if each district has at most two blocks, the problem can be solved by maximum matching in polynomial time. 

\begin{restatable}{theorem}{HardnessForPlanarGraphs}
\label{thm:planar_hard}
The $c$-balanced districting problem is $\mathsf{NP}$-hard, when $G$ is a planar graph with maximum degree $3$, each district has rank-$3$ (i.e., with at most three blocks), and the districts must be stars.
\end{restatable}

The proof of the above claim uses reduction from planar 1-in-3SAT. 
We show next that the problem on a complete graph or a tree remains hard. This reduction uses the problem of subset sum. 

\begin{restatable}{theorem}{HardnessForCompleteGraphs}
\label{thm:complete}
The $c$-balanced districting problem is $\mathsf{NP}$-hard for any $c\geq2$, when $G$ is a complete graph or a tree. This holds for both the case when the the districts are connected subgraphs and when the districts are required to be stars. 
\end{restatable}


Last, we show hardness of approximation by a reduction from the maximum independent set problem.

\begin{restatable}{theorem}{HardnessForApproximation}
\label{thm:apx}
The $c$-balanced districting problem does not have an $n^{1/2-\delta}$-approximation (for any constant $\delta>0$) in a general graph unless $\mathsf{P}=\mathsf{NP}$. On a graph with maximum degree $\Delta$, one cannot approximate the $c$-balanced districting problem within a factor of  $\Delta/2^{O(\sqrt{\Delta})}$ assuming $\mathsf{P}\neq \mathsf{NP}$, and 
$O(\Delta / \log^2 \Delta)$ assuming the Unique Games Conjecture (UGC). Even if $\Delta$ is a constant, the problem is $\mathsf{APX}$-hard. These statements hold when the districts must be stars and when the centers of the stars are limited to a subset of vertices.
\end{restatable}

\section[An Algorithm for c-Balanced Star Districting]{An Algorithm for $c$-Balanced Star Districting}\label{sec:LP-star}

In this section, we give an approximation algorithm to the $c$-balanced star districting problem.
The algorithm is based on a multiplicative weights update approach of solving packing-covering linear programs~\cite{PlotkinST91,v008a006,BKS23-lp} and then apply a randomized rounding procedure~\cite{ChanH12}.
Interestingly, the same algorithm achieves different approximation guarantees on different classes of the graphs, summarized in the following theorem.


\begin{theorem}\label{thm:star-district}
Let $G$ be a graph with function of weights $p_1$ and $p_2$.
There exists a polynomial time algorithm that computes a $c$-balanced star districting $\mathcal{T}$, with the following guarantee:
\begin{enumerate}[itemsep=0pt]
    \item[(1)] For any general graph $G$, $\mathcal{T}$ is an $O(\sqrt{n})$-approximate solution.    \item[(2)] If $G$ is planar, then $\mathcal{T}$ is an $O(\log n)$-approximate solution.
    \item[(3)] If $G$ is an $H$-minor-free graph with $|H|=h$, then $\mathcal{T}$ is an $O(h^2\log n)$-approximate solution.
    \item[(4)] If $G$ is a tree or an outerplanar graph, then $\mathcal{T}$ is an $O(1)$-approximate solution.
\end{enumerate}
\end{theorem}

In \Cref{sec:star-district-alg}, we formulate the problem as a linear program.
In \Cref{sec:star-district-algorithm-overview}, we describe how to apply the Whack-a-Mole algorithm~\cite{BKS23-lp} (with our own separation oracle) that obtains an $(1+\epsilon)$-approximate solution to the linear program in polynomial time.
In \Cref{sec:star-randomized-rounding}, we apply Chan and Har-Peled's randomized rounding technique~\cite{ChanH12}, showing that bounding the pairwise product terms leads to the desired approximation factors.
We show that there is a bound of $O(\sqrt{n})$ on any graph in \Cref{sec:star-districting-general-graphs},
To bound the pairwise product terms, we introduce the scattering separators in \Cref{sec:scattering-separator}. 
These scattering separators are useful for analyzing 
the approximation ratios for  planar graphs in \Cref{sec:planar-graphs} and 
for minor-free graphs in \Cref{sec:minor-free-graphs}.
For the graph classes that is a subclass to the planar graphs, we provide tailored-but-better analysis
 for outerplanar graphs in \Cref{sec:outerplanargraphs} and 
for trees in \Cref{sec:lp-trees}, which conclude the proof of \Cref{thm:star-district}. 

\subsection{LP Formulation}
\label{sec:star-district-alg}

We formulate the $c$-balanced districting problem as a linear program. 
For each $c$-balanced star district $S$, we define a variable $x_S$ indicating whether or not this district is chosen.
Thus, the integer linear program for 
can be defined as:
\begin{equation}\label{eq:lp-main}
\begin{aligned}
\text{maximize}\ &\sum_{S} w(S) x_S \\
\text{subject to}\ &\forall v\in V, \sum_{S\ni v} x_S \le 1\\
\ &\forall S, x_S\in \{0, 1\}
\end{aligned}
\end{equation}

To give an approximate solution to the above integer linear program, we follow the standard recipe that solves the relaxed linear program first and then apply a randomized rounding algorithm.

\paragraph{Equivalent Relaxed Linear Program.} 
In order to solve the relaxed linear program of \Cref{eq:lp-main}, we use \emph{weighted} variables: for each district $S$, we define variable $x'_S := w(S)x_S$. Hence, the equivalent linear program (and its dual linear program) we will be solving can be described as follows.
\[
\begin{array}{c|c}
(\textsc{Primal}) & (\textsc{Dual})
\\[5pt]
\begin{aligned}
    \text{maximize }\  & \sum_{S}  x'_S\\
    \text{subject to }\  & \forall v\in V,\ \sum_{S\ni v} \frac{1}{w(S)} x'_S \le 1\\
    & x'_S \ge 0
\end{aligned} &
\begin{aligned}
    \text{minimize }\  & \sum_{v} y'_v\\
    \text{subject to }\  & \forall S,\
     \sum_{v\in S}  \frac{1}{w(S)} y'_v \ge 1\\
    & y'_v \ge 0
\end{aligned}
\end{array}
\]

We note that the total number of primal variables (i.e., the number of potential $c$-balanced star districts) could be exponentially many in terms of the graph size.
However, due to the special structure of this problem,
seeking for an approximate solution does not require the participation of every variable.
We summarize the result of solving the relaxed linear program as \Cref{theorem:star-lp-main} below with details in \Cref{sec:star-district-algorithm-overview} and \Cref{sec:whack-a-mole}.
A conclusive proof to \Cref{theorem:star-lp-main} is carried out in \Cref{sec:appendix-proof-of-thm-star-lp-main}.

\begin{theorem}
\label{theorem:star-lp-main}\label{theorem:star-main}
Given a graph $G$ and a precision parameter $\epsilon \in (0, \min\{\frac12, \frac{c-2}{c}\})$, there exists an algorithm that returns an $(1-\epsilon)$-approximate solution $\{x'_S, y'_v\}$ to the above linear program in $\poly(n, 1/\epsilon, \log(w(G)))$ time.
Moreover, there are at most $\poly(n, 1/\epsilon)$ non-zero terms among the returned primal variables $\{x'_S\}$.
\end{theorem}


\subsection{Solving Relaxed Linear Program}\label{sec:star-district-algorithm-overview}
We run a modified version of Whack-a-Mole Algorithm by Bhattacharya, Kiss, and Saranurak \cite[Figure 3]{BKS23-lp} to solve the dual linear program. 
In this section we only highlight the challenges and describe the modifications that resolve the challenges.
For the sake of completeness, we 
provide a full walkthrough and the analysis of the algorithm in \Cref{sec:whack-a-mole}.

At a high level glance, the Whack-a-Mole framework first reduces  (via binary search on the objective value $\mu$ and re-scaling of the linear program) to the problem of  reporting either a feasible primal solution of value $\ge 1$ or a feasible dual solution of value $\le 1+\epsilon$. To solve this reduced task, a multiplicative weights update (MWU) approach is applied: the algorithm initializes $y'_v=1/n$, and iteratively finds any \emph{strongly violating dual constraint}. A strongly violating dual constraint refers to a $c$-balanced district $S$ such that $\sum_{v\in S} y'_v < (1-\epsilon)\cdot w(S)$. Upon discovering any strongly violating constraint, the algorithm multiplicatively increases the weight of the vertices that are involved in the constraint.
To guarantee an efficient runtime, the Whack-a-Mole algorithm uses a lazy and approximated MWU approach where it normalizes the dual variables $\{y'_v\}$ only when the sum exceeds $1+\epsilon$. 
In the end, if all constraints are satisfied, then a feasible dual solution of value $\le 1+\epsilon$ is found. Otherwise, the MWU framework guarantees a feasible primal solution after a sufficient number 
(namely, at least $\mu \epsilon^{-2}\ln n$)
of iterations.

\paragraph{The Challenges and the Solutions.}
The main obstacle from directly applying the BKS Whack-a-Mole algorithm is that their algorithm runs in near-linear time in terms of the number of non-zero entries. 
In particular, in their algorithm, every constraint is examined at least once. If a violating constraint is detected, a ``whack'' step is performed which corresponds to a MWU step and perhaps fixing that violation. 
In our case, the number of constraints of a dual linear program may be exponential in $n$.
Thus, a straightforward analysis to a direct implementation of BKS algorithm leads to a running time exponential in $n$.

Fortunately, in our case it is not necessary to explicitly check every constraint. 
We develop a polynomial time \emph{separation oracle} which either locates a good enough violating constraint or reports that all constraints are not strongly violating.
With a tweaked potential method, we are able to show that, if in each iteration, we carefully pick a good enough violating constraint, the number of violating constraints being considered during the entire Whack-a-Mole process will be at most $\text{poly}(n, 1/\epsilon)$.
Therefore, by plugging in our separation oracle to the Whack-a-Mole algorithm, we are able to claim a polynomial runtime.



\paragraph{The Requirements to the Separation Oracle.}
Assume $\mu=1$ (or we may assume that for all vertex $v$, the weights $w(v)$ are already replaced with $w(v)/\mu$) for simplicity.
To obtain a polynomial upper bound on the number of iterations,
our algorithm chooses some violating constraint that approximately maximizes its weight.
Let us define the set of all \emph{strongly violating cosntraints} as follows.
\begin{equation}\label{eq:violate}
\mathcal{S}_{\mathrm{violate}} := \left\{ S\ \bigg|\ \sum_{v\in S} y'_v < (1-\epsilon)\cdot w(S) \right\}.
\end{equation}
Our separation oracle does the following: it either returns any (possibly weaker) violating $S'$ such that 
\begin{equation}\label{eq:weak-violate}
\begin{aligned} 
v\in S' && \text{and}\ \ \ \sum_{v\in S'} y'_v < (1-\epsilon/2)\cdot w(S') && \text{and}\ \ \  w(S') \ge \frac12\cdot \max_{T\in \mathcal{S}_{\mathrm{violate}}} w(T),
\end{aligned}
\end{equation}
or certifies that $\mathcal{S}_\mathrm{violate}=\emptyset$.

\paragraph{Bounding the Number of Calls to the Separation Oracle.}
We show that the above modification leads to a polynomial upper bound on the number of iterations.
It suffices to bound the number of iterations within a single \emph{phase} --- within a phase, the dual variables are not normalized and are non-decreasing. At the beginning of a phase, the dual variables are normalized such that their sum becomes $1$.
Let us define a potential function $\phi$ for each vertex $v$. For brevity we denote $\mathcal{S}_v := \{S\in\mathcal{S}_{\mathrm{violate}}\ |\ v\in S\}$.

\begin{equation}\label{eq:def-potential}
\phi(v) := \begin{cases}
    1 -  \displaystyle\frac{y'_v}{\displaystyle  \max_{S\in \mathcal{S}_v} w(S)} & \text{if } \mathcal{S}_v\neq\emptyset,\\
    0 & \text{otherwise.}
\end{cases}
\end{equation}

Notice that, whenever there is a violating constraint $S$ that contains $v$, it is guaranteed that $\phi(v) > \epsilon$.
This can be proved by contradiction:
If $\phi(v) \le \epsilon$, then for any $S\in\mathcal{S}_v$, we have $$\sum_{u\in S} y'_u \ge y'_v \ge (1-\epsilon) \cdot \max_{T\in \mathcal{S}_v} w(T) \ge (1-\epsilon)\cdot  w(S),$$ which implies that $S\notin \mathcal{S}_{\mathrm{violate}}$ by \Cref{eq:violate}, a contradiction. 

Suppose now a violating constraint $S'$ is returned from the separation oracle and $v\in S'$.
By \Cref{eq:weak-violate}, we know that $\sum_{v\in S'} y'_v < (1-\epsilon/2)\cdot w(S')$. After ``whacking'' the constraint, all vertices in $S'$ have their dual value increased such that $\sum_{v\in S'} y'_v \ge w(S')$ (or the total sum $\sum_v y'_v$ becomes too large and a new phase is considered.)
Let $\Delta y'_v$ denote the net increase of the value $y'_v$ after ``whacking'' the constraint $S'$.
Thus,
\begin{equation}
\sum_{v\in S'}\Delta y'_v \ge {w(S')}\cdot (\epsilon/2),
\end{equation}
which implies that, there exists at least one $v\in S'$ such that
\begin{equation}\label{eq:7}
    \Delta y'_v \ge \frac{1}{|S'|} \cdot w(S')\cdot (\epsilon/2) \ge \frac{1}{n} \cdot w(S')\cdot (\epsilon/2) \ge \frac{\epsilon}{4n} \cdot \max_{S\in \mathcal{S}_v} w(S).
\end{equation}
We remark that the last inequality is due to the property of the violating constraint returned by the separation oracle (see \Cref{eq:weak-violate}).
Now, 
we claim that the new potential at this particular vertex $v$ must decrease by at least $\epsilon/(4n)$ (or drop to 0). Indeed, we notice that after each whack, $y'_v$ is non-decreasing and $\max_{S\in \mathcal{S}_v} w(S)$ is non-increasing (since the set of violating constraints $\mathcal{S}_v$ might shrink). When there are still strongly violating constraint after the update, the new potential function for $v$ becomes:
\begin{align*}
\phi_{\mathrm{new}}(v) &= 1 - \frac{y'_v + \Delta y'_v}{\displaystyle  \max_{S\in \mathcal{S}^{\mathrm{new}}_v} w(S)}
\le 1 - \frac{y'_v + \Delta y'_v}{\displaystyle  \max_{S\in \mathcal{S}_v} w(S)} \tag{$\{y'_u\}$ non-decreasing implies $\mathcal{S}^{\mathrm{new}}_v\subseteq \mathcal{S}_v$}\\
&= \left(1 -  \frac{y'_v}{\displaystyle  \max_{S\in \mathcal{S}_v} w(S)}\right)
- \frac{\Delta y'_v}{\displaystyle  \max_{S\in \mathcal{S}_v} w(S)} 
\le \phi(v) - \frac{\epsilon}{4n}. \tag{by \eqref{eq:def-potential} and \eqref{eq:7}}
\end{align*}
Therefore,
we conclude that there can only be at most $O(n^2/\epsilon)$ whacks during a phase, since there are $n$ vertices and each can be whacked at most $O(n/\epsilon)$ times.
Finally, a bound of $O(\log n/\epsilon^2)$ to the number of phases (see \Cref{thm:mwu-phases}) implies a desired polynomial upper bound.

\paragraph{Implementing the Separation Oracle.}
The last piece of the entire polynomial time algorithm would be to implement a desired separation oracle efficiently.
The following lemma (\Cref{lem:separation-oracle}) reduces the task of finding a violating district to solving the $c$-balanced star districting problem in a complete graph.

\begin{restatable}{lemma}{LemmaSeparationOracle}
\label{lem:separation-oracle}
Given an input instance $G=(V, E, (p_1, p_2))$, re-weighted values $w':V\to \{0\}\cup [\frac{1}{w(G)}, w(G)]$, and dual variables $y'_v\in [n^{-(1+1/\epsilon)}, 1+\epsilon]$ for each vertex $v\in V$, there exists an algorithm that either reports that $\mathcal{S}_\mathrm{violate}=\emptyset$, or returns a $c$-balanced district $S$ such that
$\sum_{v\in S} y'_v < (1-\epsilon/2) w'(S)$ and $w'(S)\ge \frac12 w'(S_{\mathrm{max}})$, where $S_{\mathrm{max}}=\arg\max_{S\in \mathcal{S}_\mathrm{violate}} w'(S)$ is a violating $c$-balanced district with the maximum value. The algorithm runs in $O(\epsilon^{-6} n^6(\log n)(\log w(G))^4)$ time.
\end{restatable}

Since the proof of \Cref{lem:separation-oracle} uses an FPTAS for the complete graph case (a more general version which looks for connected $c$-balanced districting is described in  \Cref{sec:complete}), we defer the proof to \Cref{sec:appendix-solve-separation-oracle}.
We remark that the proof of \Cref{lem:separation-oracle} is dedicated in the two-color and star districting setting. That said, if one would like to extend the framework to a more general setting (e.g., bounded diameter districting, general connected subgraph districting, or three-or-more colors setting), the implementation of the separation oracle has to be re-designed.

\subsection{The Randomized Rounding Algorithm}\label{sec:star-randomized-rounding}

We use a randomized rounding technique modified from \cite[Section 4.3]{ChanH12}.
Intuitively, the algorithm maintains a set $I$ of non-overlapping districts, which is initially empty, and keeps adding districts into $I$.

The rounding algorithm is described as follows.
Let $\{x_S\}$ be the output of an approximate solution to LP.
Let $\mathcal{S}_{\mathrm{LP}}=\{S\ |\ x_S\neq 0\}$ be the support of the solution.
The algorithm first 
sorts all non-zero valued districts according to their weights $w(S)$, from the largest to the smallest.
Let $\tau \ge 1$ be a parameter to be decided later.
For each district $S$, with probability $x_S / \tau$, the algorithm adds $S$ into $I$ as long as there is no district in $I$ overlapping with $S$. \footnote{The randomization step appears to be necessary since there is an example where deterministic rounding incurs a large approximation factor, see \Cref{appendix:det-rounding}. }
The algorithm outputs $I$ after all districts in $\mathcal{S}_{\mathrm{LP}}$ are considered.
The necessity of scaling the non-zero variables by $\tau$ comes from the analysis of expected total weight in $I$.
In an actual implementation of the algorithm, one can make the algorithm oblivious of $\tau$, by iteratively testing on different values of $\tau=(1+\epsilon)^k$ for $k=0, 1, 2, \ldots, O(\epsilon^{-1}\log n)$ and then picking the largest weighted districting among the returned ones.

\paragraph{Analysis.} The output $I$ of the algorithm can be seen as a random variable.
Let $w(I)$ be the total weight of the districts within $I$.
A straightforward analysis (see \Cref{lem:randomized-rounding}) shows that 
\[
\mathbf{E}[w(I)] \ge \sum_{S\in \mathcal{S}_{\mathrm{LP}}} w(S) \frac{x_S}{\tau} - \sum_{A, B\in \mathcal{S}_{\mathrm{LP}}:\ A\cap B\neq \emptyset} \min(w(A), w(B)) \frac{x_Ax_B}{\tau^2}\ .
\]

The right hand side of the above expression contains a weighted correlation term.
The technique by Chan and Har-Peled~\cite{ChanH12} transforms the above weighted correlation terms into \emph{unweighted} ones. 
They mentioned that, a desired $O(\tau)$-approximate solution can be achieved, as long as for any $\delta$-\emph{thresholded} subset $\mathcal{S}_{\ge\delta} := \{S\in \mathcal{S}_{\mathrm{LP}}\ |\ w(S)\ge \delta\}$, 
the total unweighted correlation terms between overlapped districts 
can be upper bounded by the sum over all primal variables within the subset:
\begin{equation}\label{eq:rounding-condition}
    \sum_{A, B\in\mathcal{S}_{\ge\delta}:A\cap B\neq\emptyset} x_Ax_B \le \frac{\tau}{2} 
    \cdot \sum_{S\in\mathcal{S}_{\ge\delta}} x_S\ \ \ \ \forall \delta > 0
\end{equation}
The above condition implies the following (see \Cref{lem:rounding-tau}):
\[
\mathbf{E}[w(I)] \ge \frac{1}{2\tau} \sum_{S\in\mathcal{S}_{\mathrm{LP}}} w(S)\cdot x_S \ge \frac{1}{2\tau (1+\epsilon)} \mathrm{OPT}_{\mathrm{LP}},
\]
where $\mathrm{OPT}_{\mathrm{LP}}$ is the optimal value of the LP problem. Thus, $I$ is an $O(\tau)$-approximate solution in expectation. 
We remark that due to the factor $2$ appearing in the denominator, this linear program with randomized rounding approach (although not contradicting) is unlikely to achieve a PTAS.

The above analysis to the rounding algorithm enables the approach of seeking a suitable $\tau$ value, such that 
\Cref{eq:rounding-condition} holds.
The rest of the section focuses on providing upper bounds of $\tau$ on various classes of input graphs.

\subsection[General Graphs]{An $O(\sqrt{n})$-Approximation Analysis for General Graphs}\label{sec:star-districting-general-graphs}

In this section,
we show that the algorithm achieves an $O(\sqrt{n})$-approximate ratio on any graph, by giving an upper bound $\tau=O(\sqrt{n})$  for the randomized rounding algorithm (with proof in \Cref{sec:appendix-general-graphs}). 

\begin{restatable}{lemma}{GeneralGraphRoundingGap}
\label{lem:lp-general-graph}
Let $G$ be any graph. Let $\{x_S\}$ be any feasible solution to the linear program. Then, $$\sum_{A, B:\ A\cap B\neq \emptyset} x_Ax_B\le \sqrt{n} \cdot \sum_{S} x_S.$$
\end{restatable}

\paragraph{Remarks: Integrality Gap.} 
We remark that this algorithm is achieving nearly the best approximation factor since it is $\mathsf{NP}$-hard to have approximation factor of $n^{1/2-\delta}$ for any constant $\delta>0$. Related to this, we would like to examine the potential loss in different steps of our algorithm. In the first step, we relax the integer linear program to a linear program with variables taking real numbers, the \emph{integrality gap} refers to the ratio of the optimal fractional solution to the optimal integer solution (since we consider maximization problem the optimal fractional solution is no smaller than the optimal integer solution). In the second step of our algorithm, we use randomized rounding to turn the fractional solution back to a feasible integer solution. 
We call the ratio between the sum of products between overlapping districts' primal variables and the sum of all variables to be the \emph{rounding gap}, i.e., $\sum_{A, B: A\cap B\neq\emptyset} x_Ax_B/\sum_{S} x_S$. 
The rounding gap can be used to upper bound the loss of solution quality when we turn the fractional solution to a feasible integer solution using the randomized rounding algorithm. 

Necessarily, a large integrality gap implies a large rounding gap for sure. Specifically, let $\tau$ be the rounding gap. 
The analysis to our rounding algorithm guarantees the existence of an integral solution within a factor of $4(1+\epsilon)\tau$ from the optimal fractional solution, which implies an integrality gap of at most $4\tau$ when setting $\epsilon \to 0$. Thus if the integrality gap is large, we cannot have a small rounding gap.
Interestingly, the above discussion, combined with \Cref{lem:lp-general-graph} implies that the integrality gap of the natural LP formulation for the star districting problem is at most $O(\sqrt{n})$. 

Next we show that our LP relaxation could have a large integrality gap of $\Omega(\sqrt{n})$ for a general graph. 
We use a reduction from $k$-uniform hypergraph matching problem to our $c$-balanced star districting problem.
Let $H=(V_H, E_H)$ be the given $k$-uniform hypergraph -- a hypergraph such that all its hyperedges have size $k$. We construct a graph $G=(V_H\cup E_H, E_G)$ by creating additional vertices for each hyperedge. These vertices have heavy weights, say $p_2(e) := (c-1)k$ and vertices from $V_H$ have weights $p_1(v) := 1$.    For each hypergraph $e\in E_H$ (which is a subset of vertices), we connect all vertices $v\in e$ to the corresponding vertex $e$ in $G$. It ensures that there is an one-to-one correspondence between hyperedges of $H$ to $c$-balanced star districts on $G$. 
Now, the relaxed linear program for $(G, p_1, p_2)$ will be equivalent (with an extra $ck$-factor in the objective function) to a fractional hypergraph matching. Thus, the $(k+1-1/k)$ integrality gap of $k$-uniform hypergraph matching~\cite{FurediKS93, ChanL12} can be transferred to our LP formulation --- specifically, the construction in \cite{ChanL12} via projective planes leads to an $\Omega(\sqrt{n})$ integrality gap.

On the other hand, we do observe \emph{planar graph instances} with a constant $>1$ rounding gap with an at most $1+o(1)$ integrality gap (please see \Cref{sec:large-rounding-gap}).
Again, this does not eliminate the possibility of achieving PTAS, but it suggests a conjecture that we are unlikely to obtain a PTAS for planar graphs using the current analysis.

%

\subsection{Scattering Separators}\label{sec:scattering-separator}

Let us now introduce the scattering separators, which is useful for the divide and conquer framework for upper bounding the approximation ratio of the randomized rounding procedure.

\begin{definition}
Let $G=(V, E)$ be a graph and let $X\subseteq V$ be any subset. We say that $X$ is: 
\begin{itemize}[itemsep=0pt]
\item \emph{$(k, t)$-scattered}, if $X$ can be partitioned into at most $k$ subsets $X = X_1\cup X_2\cup \cdots \cup X_k$ with each $X_i$ being a $t$-hop independent set\footnote{We say that $X$ is a \emph{$t$-hop independent set} (with respect to the graph $G$) if for all pairs of distinct vertices $u, v\in X$ and $u\neq v$, the shortest distance between $u$ and $v$ is at least $t$ on $G$.};
\item \emph{$(k, t)$-orderly-scattered}, if there exists a way to partition $X$ into a sequence of at most $k$ subsets $X = X_1\cup X_2\cup \cdots \cup X_k$, where each $X_i$ is a $t$-hop independent set after the removal of all previous subsets $G-\cup_{j<i}X_j$.
\end{itemize}
\end{definition}

\begin{definition}
Let $G$ be a graph, $k, t\in\mathbb{N}$, and $\delta\in (0, 1)$. A \emph{$(k, t, \delta)$-scattering separator} is a subset of vertices $X\subseteq V$ such that (1) $X$ is $(k, t)$-orderly-scattered, and (2) $X$ is a balanced separator of $G$, that is, the largest connected component of $G-X$ has at most $\delta n$ vertices.
\end{definition}

We remark that a $(k, t)$-scattered set is also $(k, t)$-orderly-scattered. This orderly-scattered definition are useful when we  remove  subsets of vertices sequentially --- they are used in the analysis of, for example, planar graphs and minor-free graphs. On the other hand, for some graph class such as outerplanar graphs it suffices to use $(k, t)$-scattered sets within the analysis.

The scattering separators are useful in the $c$-balanced districting problem for $t\ge 5$. To justify this, suppose that we have a $5$-hop independent set $Y$.
Any star district contains at most one vertex in $Y$. If two star districts contain different vertices of $Y$, the two districts must be disjoint. Thus we partition the pairs of overlapping districts by whether they overlap with $Y$ or not, and if so, which vertex of $Y$. The following fact can be easily verified.

\begin{fact}\label{lem:why-five-hop}
Let $Y$ be a $5$-hop independent set.
Consider a fixed district $A\in\mathcal{S}$.
Assume there is a district $B\in\mathcal{S}$ that overlaps with $A$ and $A\cup B$ touches $Y$, i.e., $A\cap B\neq\emptyset$ and $(A\cup B)\cap Y\neq\emptyset$. Since the diameter of $G[A\cup B]$ is at most $4$, we know that $|(A\cup B)\cap Y|=1$. 
Further, if $A$ overlaps with two other star districts $B, C$ with both centers of $B, C$ in $Y$, then $B$, $C$ have the same center. 
\end{fact}

Fix a district $A\in \mathcal{S}$. Since all other districts that overlap with $A$ contains (at most) the same vertex in $Y$, these primal variable values add up to at most $1$.
This implies that, removing $Y$ from $G$ charges at most one copy of $\sum x_S$.
If we are able to show that the entire vertex set is a $(k, 5)$-orderly-scattered, then we obtain a desired $\tau=O(k)$ value for \Cref{eq:rounding-condition}.
However, we do not know if such a constant $k$ can be achieved for planar graphs. 
Fortunately, using the idea of balanced separators, we are able to achieve a polylogarithmic approximate solution.


\begin{lemma}\label{lem:divide-and-conquer-analysis}
If $G$ and all its subgraphs have a $(k, 5, \delta)$-scattering separator, then the districting obtained from executing the algorithm on $G$ is a $(2k\log_{1/\delta} n)$-approximated solution.
\end{lemma}

\begin{proof}
Let $X=X_1\cup X_2\cup \cdots \cup X_k$ be a $(k, 5, \delta)$-scattering separator of $G$.
Let $\mathcal{S}$ be the set of all districts. 
Then, all summands of the form $x_Ax_B$ where $A, B\in\mathcal{S}$ and $A\cap B\neq\emptyset$ can be also split into three parts: 
\begin{enumerate}[itemsep=0pt]
\item[(1)] $X\cap \{c_A, c_B\}\neq \emptyset$: one of the centers $c_A$ or $c_B$ is in $X$.
\item[(2)] $X\cap \{c_A, c_B\}= \emptyset$ but $X\cap A\cap B\neq \emptyset$: one of their common vertices is in $X$.
\item[(3)] None of the above.
\end{enumerate}
For $j\in \{1, 2, 3\}$, we denote $\mathit{cost}_j$ the sum of products of those overlapping districts of case $(j)$.
For case (1), 
using the given constraint that
$X$ is $(k, 5)$-orderly-scattered, we consider removing each set $X_i$ one at a time from the graph in the increasing order of $i$.
For each $X_i$, 
without loss of generality, we may swap the role of $A$ and $B$ such that for each summand we have $c_B\in X$.
By applying \Cref{lem:why-five-hop} (with $Y=X_i$), we know that for each district $A\in\mathcal{S}$, 
all districts $B$ that overlap with $A$ with $c_B\in X_i$ are actually centered at the same vertex.
This implies that the sum of all such $x_B$ values will be at most $1$ by the primal constraint. Hence, the contribution of any district $A\in\mathcal{S}$ under case (1) for $X_i$ in the graph $G-\cup_{j<i} X_j$ is at most $$\sum_{B: A\cap B\neq\emptyset \text{ and } c_B\in X_i} x_Ax_B\le x_A.$$

By summing over all $A\in\mathcal{S}$ and over all the $k$ sets $X_1, \ldots, X_k$, we have
$\mathit{cost}_1\le k\cdot(\sum_{S} x_S)$.

For case (2), the terms can be partitioned according to the common vertex $c$:
\begin{align*}
\mathit{cost}_2 &\le \sum_{j=1}^k\sum_{c\in X_j}\sum_{c\in A\cap B} x_Ax_B 
\le \sum_{j=1}^k\sum_{c\in X_j} \left(\sum_{c\in A} x_A\right)^2 
\le \sum_{j=1}^k\sum_{c\in X_j} \sum_{c\in A} x_A \le k\cdot\left(\sum_{S} x_S\right).
\end{align*}
Again here we use the property that for any fixed vertex $c$, the sum of the primal variables for star districts containing $c$ sum up to be at most $1$, i.e., $\sum_{c\in A} x_A\le 1$. Further, fix an $X_i$, any star district includes at most one vertex from $X_i$. 

For case (3) we can delegate the cost to the recursion. Notice that, all districts whose centers are in $X$ will not participate in case (3). Hence, when considering each of the connected component in $G-X$, all the districts (after chopping off vertices in $X$) are still connected and are star-shaped.

Since $X$ is a balanced separator, the divide and conquer analysis has at most $\log_{1/\delta} n$ layers. Thus, the sum over all products of overlapping districts is bounded by at most $2k\log_{1/\delta} n$ times the sum $\sum_S x_S$.
\end{proof}


\subsection{Planar Graphs}\label{sec:planar-graphs}



We are now able to derive an $O(\log n)$-approximation bound for planar graphs, immediately followed by the lemma below and \Cref{lem:divide-and-conquer-analysis}.

\begin{lemma}\label{lem:planar-star}
Every planar graph has a $(10, 5, 2/3)$-scattering separator.
\end{lemma}

The above separator lemma can be derived from
the fundamental cycle separators, which is composed of two shortest paths on a BFS tree:

\begin{fact}[\cite{Lipton1979-kk}]\label{fact:fundamental-cycle-separator}
Let $G=(V, E)$ be a planar graph. Then, there exists a partition of $V=L\cup X\cup R$, such that (1) both $|L|, |R|\le \frac23|V|$, (2) no edges connect between $L$ and $R$, and (3) the separator $X$ is formed by the union of two root-to-node paths from some BFS tree on $G$.
\end{fact}

\begin{proof}[Proof of \Cref{lem:planar-star}]
With \Cref{fact:fundamental-cycle-separator}, we know that there exists a separator $X$ that is a union of two shortest paths $X=P_1\cup P_2$. Since each $P_i$ is a shortest path, we can partition each path into at most 5 sets and each set is a 5-hop independent set. Specifically, along a shortest path $P_i$, we color the vertices sequentially by color $1$ to $5$ in a round-robin manner. Each set of vertices of the same color is a $5$-hop independent set. With two shortest paths, we have a total of $10$ such $5$-hop independent sets.
Thus, each planar graph has a $(10, 5, 2/3)$-scattering separator.
\end{proof}

\subsection{Minor-Free Graphs}\label{sec:minor-free-graphs}

A natural generalization of planar graph would be minor-free graphs.
Let $H$ be a graph with $h$ vertices. 
Since planar graphs are $\{K_5, K_{3, 3}\}$-free which implies $K_6$-free, one may expect the rounding algorithm achieves a similar approximation ratio bound for $H$-minor-free graphs as well.

Indeed, in this section,
we show that our randomized rounding algorithm achieves an $O(h^2\log n)$-approximation ratio for $H$-minor free graphs.
Our proof is inspired by the seminal separator theorem of Alon, Seymour, and Thomas~\cite{AlonST90} for $H$-minor free graphs.
But, our algorithm is substantially simpler since we do not require the separator to be small in size. Instead, all we need is to obtain a balanced separator $X$ which is a union of at most $O(h^2)$ $5$-hop independent sets.
By removing these independent sets one after another, we are able to bound the correlation terms that involve any vertex in $X$ by $O(h^2)\cdot \sum x_S$.
Finally, an additional $\log n$ factor will then be added to the approximation ratio by the divide and conquer framework, achieving an $O(h^2\log n)$ approximation ratio.

Since $K_h$-minor-free implies $H$-minor-free for any graph $H$ of $h$ vertices, the following separator theorem on $K_h$-minor-free graphs implies for all $H$-minor-free graphs.

\begin{lemma}\label{lem:h-minor-free}
Let $h\in \mathbb{N}$ be a constant. 
Let $G$ be a $K_h$-minor-free graph. Then, 
$G$ admits a $(5h^2, 5, 1/2)$-scattering separator.
\end{lemma}

\begin{proof}
Consider an algorithm that \emph{seeks for an $K_h$-minor}.
The algorithm maintains two objects on the graph: a collection $\mathcal{C}:=\{C_1, C_2, \ldots, C_k\}$ of disjoint vertex subsets, and an \emph{active subgraph} $B$, such that:
\begin{enumerate}[itemsep=0pt]
    \item For all $i$, the induced subgraph $G[C_i]$ is connected.
    \item For all $i\neq j$, $C_i$ and $C_j$ are \emph{neighboring}, that is, there exists an edge in $G$ connecting some vertex from $C_i$ and some vertex from $C_j$. If we contract every $C_i$ into a vertex in $G[\cup_i C_i]$, we obtain a complete graph $K_k$. Hence, the algoritm maintains a certificate of $K_k$-minor\footnote{In \cite{AlonST90}, they called $\mathcal{C}$ a \emph{covey}.} in $G$.
    \item Let $X=\cup_i C_i$ be the separator of our interest. Whenever $|B|>n/2$, the active subgraph $B$ is always the largest connected component\footnote{In \cite{AlonST90}, they called each connected component in $G-X$ an $X$-flap.} in $G-X$. We note that the largest connected component is uniquely defined whenever $|B| > n/2$. 
    \item Each $C_i$ is $(5h, 5)$-orderly-scattered.
\end{enumerate}

We give a high level description to our algorithm.
The algorithm initially sets $\mathcal{C}=\emptyset$ and $B=G$. In each iteration, the
algorithm repeatedly either finds a connected subgraph $C_{\text{new}}\subseteq V(B)$ that is neighboring to all $C_1, \ldots, C_k$, or removes some part $C_i$ that is not neighboring to $B$.
The algorithm halts once $|B|\le n/2$.

Now we describe the detail of each iteration. If $|B|\le n/2$ then we are done, as $X$ is a $(5hk, 5, 1/2)$-scattering separator with $k\le h$.
Suppose that $|B| > n/2$.
The algorithm would first check if there is a subset $C_i$ that is not neighboring to $B$. 
If so, the algorithm removes $C_i$ from $\mathcal{C}$.
The invariants 1, 2, and 4 clearly holds for $\mathcal{C}-C_i$ so it suffices to show that the third invariant holds as well. 
Since $B$ has no neighbor in $C_i$ and $B$ is a connected component, we know that $B$ is still a connected component of $G - (\cup_{j\neq i} C_j)$. The fact that $|B|>n/2$ implies that all connected components other than $B$ has at most $n/2$ vertices, so the third invariant holds.

Now we can assume that all subsets $C_1, C_2, \ldots, C_k$ in $\mathcal{C}$ have neighboring vertices in $B$. In this case, 
we claim that a new set $C_{\text{new}}\subseteq V(B)$ can be found such that we update $\mathcal{C}\leftarrow \mathcal{C}\cup \{C_{\text{new}}\}$ and that the invariant 1, 2, and 4 holds for the updated $\mathcal{C}$. To construct $C_\text{new}$, we first let $a_1, a_2, \ldots, a_k$ be vertices in $B$ such that for all $i$, $a_i$ has a neighbor in $C_i$. Since $B$ is a connected component, there exist paths connecting $a_i$ to $a_{i+1}$ in $B$, for all $i$, $1\le i \le k-1$. Let $P_i$ be the \emph{shortest path} from $a_i$ to $a_{i+1}$ in $B$.
We define $C_\text{new} := \cup_{i=1}^{k-1} V(P_i)$.
Invariants 1 and 2 are clearly satisfied for $C_\text{new}$.
The fact that $P_i$ is a shortest path implies that $V(P_i)$ can be partitioned into at most $5$ 5-hop independent sets. Thus, $C_\text{new}$ can be partitioned into at most $5h$ 5-hop independent sets, thereby maintaining invariant 4.
Since $G$ is $K_h$-minor-free, we know that $|\mathcal{C}| \le h-1$.

After constructing the set $C_\text{new}$ and updating $\mathcal{C}\gets \mathcal{C}\cup \{C_\text{new}\}$, the algorithm also updates the active subgraph $B$, by substituting $B$ with the largest connected component $B'$ of $G-\cup_{C_i\in \mathcal{C}} C_i$. Notice that, if $|B'| \le n/2$ then $X=\cup_{C_i\in \mathcal{C}} C_i$ is already a balanced separator. Otherwise, the fact that $B$ is the unique connected component with $|B|>n/2$ implies that $B'\subseteq B - C_\text{new}$. Hence, in any case we have $|B'| < |B|$. That is, the active subgraph strictly decreases in size.
The algorithm must halt, since in each iteration the quantity $|\mathcal{C}| + 2|B|$ is nonnegative but strictly decreasing -- either we drop a set in $\mathcal{C}$, thus $|\mathcal{C}|$ drops but $|B|$ remains the same, or we increase $|\mathcal{C}|$ by $1$ but decrease $|B|$ by $1$ at least.
Therefore, whenever the algorithm halts, we obtain a separator $X=\cup_i C_i$ as desired, which implies \Cref{lem:h-minor-free}.
\end{proof}

\paragraph{Remarks.}
If we apply \Cref{lem:h-minor-free} to planar graphs, we may use the fact that planar graphs are $K_6$-minor-free and thus obtaining a $(180, 5, 1/2)$-scattering separator. This leads to an $(360\log_2 n)$-approximation ratio by applying \Cref{lem:divide-and-conquer-analysis}.
Note that this is much worse than $10\log_{3/2}n\approx 17.1 \log_2 n$, which is the bound we obtained from \Cref{lem:planar-star}.

\subsection{Outerplanar Graphs}\label{sec:outerplanargraphs}

The results from previous subsections suggest approximation ratios that involve polylogarithmic terms.
It is, of course not surprising, possible to obtain a better approximate ratio for subclasses of planar graphs.
In this subsection, we show that our algorithm produces an $O(1)$-approximate solution for outerplanar graphs, using a slightly different approach.

\begin{restatable}{lemma}{OuterPlanarGraphRoundingGap}
\label{lem:outerplanar}
Let $G$ be an outerplanar graph. Suppose that $\{x_S\}$ are primal variables obtained by \Cref{theorem:star-main}. Then, $\sum_{A\cap B\neq\emptyset} x_Ax_B\le O(1)\cdot \sum x_S$.
\end{restatable}

We defer the proof of \Cref{lem:outerplanar} to \Cref{sec:appendix-outerplanar-graphs}.
For an even simpler graph classes such as trees, we are able to obtain a much smaller constant bound. We provide such an analysis in \Cref{sec:lp-trees}.

\section{FPTAS for General Districting on Complete Graphs and Trees}
\label{sec:FPTAS}

In this section, we present FPTAS for complete graphs and trees with weighted blocks. The algorithms here find $c$-balanced, connected districts that can be more than a star graph. 
Further, for complete graphs and trees, the LP-based algorithm in the previous section achieves $O(1)$-approximation ratio while the algorithms in this section achieves a ratio of $1+\epsilon$.

\subsection{Complete Graph}\label{sec:complete}
Let $G$ be a complete graph with  functions of weights $p_1$ and $p_2$.   
Because we can merge two adjacent $c$-balanced districts on $G$ into a single $c$-balanced district as shown in \cref{lemma:mergeable-property}, the $c$-balanced districting problem on complete graphs can be reduced to obtaining \emph{one} $c$-balanced district, described as the following:

\begin{fact}[Mergeable Property]\label{lemma:mergeable-property}
Assume $T_1$ and $T_2$ are disjoint districts and $G[T_1\cup T_2]$ is connected. If $T_1$ and $T_2$ are both $c$-balanced, then $T_1\cup T_2$ is also a $c$-balanced district.
\end{fact}

\begin{mdframed}
\textsc{Complete-Graph-$c$-Balanced-Districting}\\
\textbf{Input:} Let $G = (V, E)$ be a complete graph of $n$ blocks and function of weights $p_1$ and $p_2$.\\
\textbf{Goal:}
Obtaining a subset $S\subseteq V$ such that the total weight $w(S)$ is maximized subjected to the $c$-balanced condition:
\begin{align}
(c-1)p_1(S)-p_2(S)&\ge 0 & \text{and} && (c-1)p_2(S)-p_1(S)&\ge 0.\label{eq:complete_balanced}
\end{align}
\end{mdframed}

The following theorem gives an FPTAS using dynamic programming (\cref{alg:complete}).

\begin{algorithm}[htp]
\caption{FPTAS on complete graphs}\label{alg:complete}
\KwData{$\epsilon>0$, $c> 2$, a complete graph $(V,E)$, $V = \{v_1, \dots, v_n\}$, functions of weights $\vp = (p_1, p_2)$}
\SetKwFunction{Ftrim}{Trim}
  \SetKwProg{Fn}{Function}{:}{}
  \Fn{\Ftrim{$L$, $\ell$, $\epsilon$}}{
        Sort $L = \{\vq_1, \dots, \vq_m\}$ so that $\ell(\vq_1)\ge \ell(\vq_2)\ge\dots\ge  \ell(\vq_m)$\;
        Set $L_{out} = \emptyset$\;
        \For{$i = 1,\dots, m$}
        {
            \If{$\vq_i$ is not marked}{
            $L_{out}\gets L_{out}\cup \{\vq_i\}$\;
            Mark all $\vq_j\in L$ that $\epsilon$-approximates $\vq_i$\;
            }
        }
        \KwRet $L_{out}$\;
  }

Set $L_1^0 = L_2^0 = \{(0,0)\}$\;
\For{$i = 1,\dots, n$}
{
    $L_1^i \gets \Ftrim{$L_1^{i-1}\cup (L_1^{i-1}+\vp(v_i)), \ell_1, \epsilon/n$}$\;
    $L_2^i \gets \Ftrim{$L_2^{i-1}\cup (L_2^{i-1}+\vp(v_i)), \ell_2, \epsilon/n$}$\;
}
\Return the largest $c$-balanced districting in $L_1^n\cup L_2^n$\;
\end{algorithm}



\begin{restatable}{theorem}{CompleteGraphDistrictingTheorem}
\label{thm:fptas-complete}
There exists an FPTAS algorithm solving \textsc{Complete-Graph-$c$-Balanced-Districting} so that for all $c> 2$, $0<\epsilon< \frac{1}{2}\ln(c-1)$, and complete graph $(V,E)$ of $n$ nodes with functions of weights $\vp = (p_1, p_2)$, the algorithm outputs an $e^\epsilon$-approximation in $O\left(\epsilon^{-4}n^6 (\ln w(V))^4\right)$ time where $w(V) = \sum_{v\in V} p_1(v)+p_2(v)$.
\end{restatable}



We defer a detailed proof of \Cref{thm:fptas-complete} to \Cref{sec:omitted-proofs-from-fptas} and here give a high level idea.  One naive approach involves creating a complete list of potential subset sum values, denoted as $L(V)$ and outputting the largest $c$-balanced one.  While this approach finds an optimal solution, it is not necessarily efficient, as $L(V)$ can be exponentially large.  
Similar to the knapsack problem or subset sum problem, one may use a bucketing idea to trim the list, keeping only one value when several are close to each other.
However, the $c$-balanced constraint posts a challenge for the algorithm --- for example, if an trimming algorithm keeps partial districts during the iterations, these partial districts may not always remain $c$-balanced resulting in a poor approximation ratio. 
To address this, we design a prioritized trimming process that prioritizes subsets satisfying the $c$-balanced condition in \cref{eq:complete_balanced} such that any $c$-balanced district in $L(V)$ would have an approximated district in our trimmed list. 

Specifically, given $\epsilon\ge 0$, we say $\vq$ is an \emph{$\epsilon$-approximate} of $\vq'$ if 
$q_1/q_1', q_2/q_2'\in [e^{-\epsilon}, e^{\epsilon}]$ where $0/0 := 1$. Let
$\ell_1(\vq) = (c-1)q_1-q_2$ and $\ell_2(\vq) = (c-1)q_2-q_1$ be two linear functions on $\vq = (q_1, q_2)\in \mathbb{R}^2$.  We say  $\vq$ is \emph{$\ell_j$-dominated} by $\vq'$ if $\ell_j(\vq)\le \ell_j(\vq')$ for $j = 1,2$, and $\vq, \vq'\in \mathbb{R}^2$, and $L'$ is a \emph{$(\ell_j,\epsilon)$-trimmed} of $L$ if $L'\subseteq L$ and for each $\vq\in L$ there exists $\vq'\in L'$ which $\epsilon$-approximates and $\ell_j$-dominates $\vq$.  
The key observation is that if $\vq$ is $c$-balanced satisfying \cref{eq:complete_balanced} with $q_2\ge q_1$ and $\vq'$ $\ell_1$-dominates $\vq$, $\vq'$ is also $c$-balanced.  A similar argument holds for $q_1\ge q_2$. This observation suggests that when trimming multiple nearby values, we keep the one that optimizes $\ell_1$ (and $\ell_2$) that ensures the existence of $c$-balanced approximated values.  Therefore, we can find an $e^\epsilon$-approximated solution if we can compute $(\ell_j, \epsilon)$-trimmed of all possible subset sum values $L(V)$.  
Moreover, because $\ell_1$ and $\ell_2$ are linear, we can use dynamic programming to sequentially and efficiently compute $L^i_1$ and $L^i_2$ that is  $(\ell_1, \frac{\varepsilon i}{n})$-trimmed and  $(\ell_2, \frac{\varepsilon i}{n})$-trimmed of all possible subset sum values on the first $i$ blocks $L^i = L(\{v_1, \dots, v_i\})$ respectively.  
While \cref{alg:complete} only returns the size of our approximated solution $\vq\in \mathbb{R}^2$, we can use an additional $n$ factor to store the set $S$ for each $\vq$ 
in $L^i_1, L^i_2$ to recover our approximated optimal districting. 

Finally, we note that our prioritized trimming that ensures both inequalities in \cref{eq:complete_balanced}: one through prioritized $\ell_j$ the other through exhausting cases of $q_2^*\ge q_1^*$ or $q_2^*\le q_1^*$. However, we cannot extend this approach to non-binary color settings.  Instead, if we allow relaxing $c$-balanced constraint to $c'$-balanced district with $c'$ slightly larger than $c$, the standard bucketing algorithm mentioned above can directly work even for the non-binary color setting.
\subsection{Tree Graph}
Similar to our dynamic programming algorithm for complete graphs, we can design an FPTAS for trees.  First, note that if we only need to find one $c$-balanced district, we can easily adapt our FPTAS for a complete graph.  Concretely, we recursively grow (incomplete) districts for a block's children in a depth-first search (DFS) order and decide whether to continue growing the district or not.  For more than one district, we can use additional memory to store the total weight of all $c$-balanced districts that are fully contained in the subtree of each block and provide an approximate  districting solution as described in \cref{alg:tree}.

\begin{restatable}{theorem}{TreeDistrictingTheorem}
\label{thm:tree}
There exists an FPTAS for the $c$-balanced districting problem on trees. That is, for all $c> 2$, $0<\epsilon< \frac{1}{2}\ln(c-1)$, and tree graph $(V,E)$ of $n$ blocks with population functions $\vp$, the algorithm outputs an $e^\epsilon$-approximation in $O\left(\epsilon^{-6}n^8 (\ln w(V))^6\right)$ time.
\end{restatable}

We provide a detailed proof of \Cref{thm:tree} in \Cref{sec:omitted-proofs-from-fptas} and provide the intuition of proof here.  We need additional notations for our algorithm.
Given a districting $\mathcal{T}$ on a rooted tree $(V,E)$, we define a sequence of \emph{stamps} $\vs^v(\mathcal{T}) = (s_1^v(\mathcal{T}), s_2^v(\mathcal{T}), s_3^v(\mathcal{T}))\in \mathbb{R}^3$ for $v\in V$ where the first two coordinates store information about the current incomplete district; and the third coordinate is the total weight of completed $c$-balanced districts in the subtree of $v$
\begin{align*}
        s_3^v(\mathcal{T}) = \sum_{\mathclap{\substack{T\in \mathcal{T}: T \text{ is in the sub-tree of } v}}}w(T).
\end{align*}
We define $s_1^v(\mathcal{T})$ and $s_2^v(\mathcal{T})$ according to the following three cases. 
\begin{enumerate}
        \item If $v$ is not in any district of $\mathcal{T}$, we call $v$ \emph{absent} and set $s_1^v = s_2^v = 0$.
        \item If there exists a district $T\in \mathcal{T}$ that is fully contained in the subtree of $v$ and $v\in T$, we call $v$ is \emph{consolidating} and set $s_1^v = s_2^v = 0$. 
        \item Finally, if $T\in \mathcal{T}$ is not fully contained in the subtree of $v$ and $v\in T$, we call $v$ \emph{incomplete} (for $T$ or $\mathcal{T}$).  Additionally, we let $T^v\subsetneq T$ be the subtree of $T$ consisting of $v$'s descendant including $v$, and set $s_1^v = p_1(T^v)$ and $s_2^v = p_2(T^v)$ which are the population of each community in the incomplete district $T^v$.   
\end{enumerate}
Note that given $\mathcal{T}$ and $v$, if $v$ is incomplete for $T\in \mathcal{T}$, $s_1^v(\mathcal{T}) = s_1^v(\{T\}) = p_1(T^v)$ and $s_2^v(\mathcal{T}) = s_2^v(\{T\}) = p_2(T^v)$.

Our FPTAS (\cref{alg:tree}) uses a prioritized trimming process similar to \cref{alg:complete} and keeps stamps of all possible $c$-balanced districtings in the depth-first search
(DFS) order: 
\begin{equation}\label{eq:tree_stamp}
    L^v = \{\vs^v(\mathcal{T}): \mathcal{T} \text{ is a $c$-balanced districting}\}.
\end{equation}
Similarly, while \cref{alg:tree} only returns the size of our approximated solution, we can use an additional $n$ factor to store one $\mathcal{T}$ for each $\vs$ in $L^v_1, L^v_2$ to recover our approximated optimal districting.

\paragraph{Adapting to the Star Districting Setting.} We note that this dynamic programming approach also works for star districts on tree graph and yields an FPTAS.  Similar to the arbitrary districts setting, we consider three cases for each $v$: the absent case where $v$ is not included in any district; the consolidating where $v$ is in a star district that is contained in its descendants; the incomplete case where $v$ is the center of a star district that is incomplete.   

\section{Greedy Algorithms for Special Settings}\label{sec:variants}

In this section we discuss greedy strategies for various special cases of the $c$-balanced districting problem.

\subsection{General graphs with bounded rank districts
}

Given $G = (V,E)$ with  functions of weights $p_1$ and $p_2$ and integer $k$, the problem of $c$-balanced districting with rank $k$  asks for a districting $\mathcal{T}$ such that each $T \in \mathcal{T}$ is a $c$-balanced district and $|T| \le k$. Further, the total weight $w(\mathcal{T}) := \sum_{T\in\mathcal{T}} (p_1(T) + p_2(T))$ should be maximized.

We first provide a reduction from the most general $c$-balanced districting problem to maximum weighted matching on hypergraphs.
In a hypergraph $H=(V_H, E_H)$, $V_H=V$ and each hyperedge $h\in E_H$ is a subset of $V$. The \emph{rank} of $h$ is defined as $|h|$ and the rank of $H$ 
is the maximum rank among all $h\in E_H$.
If all hyperedges have the same rank $r$, we say that $H$ is \emph{$r$-uniform}.
A hypergraph matching $M\subseteq E_H$ on $H$ is a collection of pairwise disjoint hyperedges. 
For a given weight function $w:E_H\to \mathbb{Z}_{\ge 0}$, the maximum weighted hypergraph matching problem seeks for a matching $M$ that maximizes $w(M) := \sum_{h\in M} w(h)$.


\paragraph{Reduction.}
The reduction is straightforward: we treat $c$-balanced districts on $G$ as hyperedges in $H$.  Specifically, for each $c$-balanced district $T$ with $|T|\le k$, we create a hyperedge $h=T$ and add $h$ to the hypergraph. We can even enforce $h$ to have rank-$k$, by adding dummy vertices to the hypergraph and $h$. 
The weight $w(h)$ is then defined naturally as the sum of vertex weights within $h$, namely $\sum_{v\in h} w(v)$. The dummy vertices have weight $0$.

Now, it is clear that any matching $M$ in $H$ has a one-to-one correspondence to a valid bounded-rank-$k$ districting $\mathcal{T}$ on $G$.
Therefore, any algorithm solving the weighted hypergraph matching problem on a $k$-uniform hypergraph $H$ can be applied to solving the $c$-balanced districting problem with bounded rank $k$.



\paragraph{Simple Case: $k=2$.} The problem become polynomial time solvable by regarding the weighted hypergraph matching problem as a maximum weighted matching problem on normal graphs. We summarize the result below.

\begin{theorem}
\label{thm:bounded_k_2}
    There exists an algorithm that given an instance of $c$-balanced districting problem with rank $k=2$ on $G = (V,E)$ with population functions $p_1$ and $p_2$ outputs the solution in polynomial time.
\end{theorem}

\begin{proof}
    Following the reduction, first consider the 2-uniform hypergraph $(V_H,E_H)$ formed using valid districts as hyperedges. $|V_H|=n$ and $|E_H|=O(n^2)$. Since the hypergraph is 2-uniform, the problem now reduces to maximum weight matching on general graphs. The fact that $H$ can be formed in linear time and we can solve maximum weight matching in $O(|E_H|\sqrt{|V_H|}\log (|V_H|\cdot w(G)))$ time \cite{DuanPS18,GabowT91} completes the proof.
\end{proof}

\paragraph{General Case: $k>2$.}
For $c$-balanced districting with rank $k$ for any $k$, we give a greedy algorithm that achieves a $k$-approximation ratio. Consider the greedy strategy described below in \Cref{alg:hypergraph_matching}.

\begin{algorithm}
\caption{A Polynomial Time Greedy Algorithm for Hypergraph Matching}
\label{alg:hypergraph_matching}
    \KwData{Given $H=(V_H, E_H)$ with the weight function $w$.}
    Initialize an empty matching $M = \emptyset$\;
    \While{$E_H\neq \emptyset$}{
    Pick $h \in E_H$ with the largest weight (breaking ties arbitrarily) and add $h$ to $M$.\\
    Remove all edges $h'$ with $h'\cap h\neq\emptyset$ from $E_H$.
    }
    \Return $M$
\end{algorithm}

\begin{lemma}\label{lem:matching}
For a given hypergraph $H=(V_H, E_H)$ with rank at most $k$, \Cref{alg:hypergraph_matching} returns a $k$-approximate solution to the maximum weight hypergraph matching problem.
\end{lemma}

\begin{proof}
Given $(V_H, E_H)$, let $M^*$ be the optimal solution, and $M$ be the output of \Cref{alg:hypergraph_matching}. 
Consider the mapping $f: M^*\to M$: for each $h\in M^*$, $f(h)$ gives the edge with the largest weight that shares at least one vertex with $h$,  i.e.
$$f(h) = \argmax_{h'\in M: h'\cap h\neq \emptyset}w(h').$$
Note that the set $\{h'\cap h\neq \emptyset\}$ is nonempty and thus $f$ is well-defined:  otherwise, $h$ itself must be added to $M$ by the greedy algorithm. Furthermore, since the greedy algorithm considers districts with a larger weight first, we have $w(h) \le w(f(h))$. 
Finally, because all edges in $M^*$ have disjoint sets of vertices, each $h'\in M$ has at most $k$ edges in the preimage $f^{-1}(h')$. Thus, 
$$w(M^*) = \sum_{h\in M^*} w(h)\le \sum_{h\in M^*} w(f(h)) \le k\sum_{h'\in M} w(h')= k\cdot w(M).$$
Therefore, $M$ is a $k$-approximate solution.
\end{proof}

Using \Cref{lem:matching}, we obtain an approximation algorithm for our $c$-balanced districting with rank $k$. 

\begin{theorem}
\label{thm:general_k}
There exists an algorithm that given an instance of $c$-balanced districting problem with rank $k$ for $k\geq3$ on $G$ with functions of weights $p_1$ and $p_2$ outputs a $k$-approximate solution in $\tilde{O}(k n^k)$ time.
\end{theorem}

\begin{proof}
    Consider the rank $k$ hypergraph $H$ obtained from $G$ such that each hyperedge in $H$ corresponds to a $c$-balanced district. Following the reduction and \cref{lem:matching}, we have that \cref{alg:hypergraph_matching} returns a $k$-approximate solution to the $c$-balanced districting problem with rank $k$. The set $H$ can be constructed in time $O(n^k)$, sorting and running the greedy algorithm takes additional $\tilde{O}(k n^k)$ time.
\end{proof}



\subsection{General graphs with bounded degree and star districts}

In this section we consider $c$-balanced districting problem on a graph $G$ with arbitrary weights that has maximum degree bounded by $\Delta$.
Since each $c$-balanced district on $G$ now has at most $\Delta+1$ blocks, by \Cref{thm:general_k} we obtain an algorithm that produces a $(\Delta+1)$-approximate solution  in $\tilde{O}(n\Delta 2^\Delta)$ time. 

\paragraph{A Slightly Improved Approximation Ratio.} 
We realized that, using a special property of star districts, we are able to achieve a (very slightly) improvement to the approximation ratio, from $\Delta+1$ to $\Delta + \frac{1}{\Delta}$.
The observation is as follows. Suppose that a tight $(\Delta+1)$-approximated solution is returned from the greedy algorithm (via \Cref{alg:hypergraph_matching}). Let $T$ be one of the returned district with $|T|=\Delta+1$.
Then, because $w(T)$ is charged for $\Delta+1$ times in the analysis, every vertex of $T$ belongs to \emph{distinct} districts in an optimal solution.
In the star districting setting, we know that the center vertex $c_T$ must form a solo $c$-balanced district and occur in the optimal solution.

The above observation leads to a local-swap idea: if a sufficiently heavy center vertex $c_T$ of a district $T$ that is about to be added to the solution by the greedy algorithm, we can simply add $\{c_T\}$ to the output instead of adding $T$ (which potentially blocks lots of other districts). We implement this idea in \Cref{alg:greedy_bounded_degree}, where the only difference is highlighted as ``special case''.

\begin{algorithm}
\caption{Greedy algorithm for star districts on bounded-degree graphs}
\label{alg:greedy_bounded_degree}
\KwData{Input graph $G = (V,E)$ with maximum degree $\Delta$, $c\ge2$.}
\KwResult{A collection of $c$-balanced star districts $\cT = \{T_1, \cdots , T_\ell\}$}
Initialize $\cT = \emptyset$, $\mathcal{S} = \emptyset$.\\ 
\For{$v \in V$}{
        Add all $c$-balanced star districts centered at $v$ (by enumeration) to $\mathcal{S}$.
    }
Sort districts in $\mathcal{S}$ by their weights in the decreasing order.\\
\While{$\mathcal{S}$ is not empty}
{
   Let $T_i$ be the largest weight district in $\mathcal{S}$ and let $c_i$ be the center of $T_i$.\\
   \If(\tcp*[h]{special case}){$\{c_i\}$ is balanced and $|T_i|=\Delta+1$ and $w(c_i)\ge   w(T_i)/\Delta$}{
        Add $\{c_i\}$ to $\mathcal{T}$.\\
        Remove all districts that contains $c_i$ from $\mathcal{S}$.
   }
   \Else{
        Add $T_i$ to $\mathcal{T}$.\\
        Remove all districts that overlap with $T_i$ from $\mathcal{S}$.
   }
}
\Return $\mathcal{T}$
\end{algorithm}

\begin{theorem}
\label{thm:bounded_degree}
    There exists an algorithm that given a graph $G=(V,E)$ with maximum degree bounded by $\Delta$ and population functions $p_i(v)\in \mathbb{Z}_{\ge 0}$ for $v \in V$ outputs a $(\Delta+\frac{1}{\Delta})$-approximate solution to the $c$-balanced districting problem with star districts in $\widetilde{O}(n\Delta2^\Delta)$ time.
\end{theorem}

\begin{proof}
    The idea is similar to proof of \cref{lem:matching}. Let $\cT^* = \{T_1^*, \cdots, T_{\ell^*}^*\}$ be any optimal districting and $\cT$ be the districting returned by \cref{alg:greedy_bounded_degree}. Let $f:\cT^* \rightarrow \cT$ be the function that for $T_i^* \in \cT^*$ returns a district in $\cT$ with largest weight that intersects with $T_i^*$ in at least one vertex.
    If the special case does not occur, we know that $w(T_i^*) \le w(f({T_i^*}))$. Otherwise, we know by the condition of the special case, we have
    $w(T_i^*) \le \Delta \cdot w(f({T_i^*}))$.
    A naive bound gives $\sum_{T_i^*\in \cT^*} w(T_i^*) \le \sum_{T\in \cT} w(T) |f^{-1}(T)|$ -- in the worst case $|f^{-1}(T)|$ can be as large as $\Delta+1$.

    Now, to provide a better approximation ratio bound, 
    it suffices for us to bound for each $T\in \cT$, how many times $w(T)$ is charged by the districts in the optimal solution.  We say that $T$ is \emph{triggered} by the special case, if $T$ is a solo $c$-balanced district that joins $\cT$ according to the special case in \Cref{alg:greedy_bounded_degree}. There are three cases:
    \begin{itemize}[itemsep=0pt]
    \item \textbf{Case 1:} $|f^{-1}(T)|\le \Delta$ and $T$ is not triggered by the special case. In this case, the contribution of $T$ to the sum would be $w(T)\cdot |f^{-1}(T)|\le \Delta\cdot w(T)$.
    \item \textbf{Case 2:} $|f^{-1}(T)|=1$ and $T=\{c_i\}$ is triggered by the special case. In this case, when $T$ joins the output set $\cT$, the district that the algorithm is actually considering has weight at most $\Delta\cdot  w(T)$. Hence, the district in the optimal solution $T^*\in f^{-1}(T)$ must satisfy $w(T^*)\le \Delta \cdot w(T)$.
    \item \textbf{Case 3:} $|f^{-1}(T)|=\Delta+1$ and $T$ has not been triggered by the special case.
    In this case, there are exactly $\Delta+1$ districts in the optimal solution. Using the observation, the center vertex $c_T$ forms a solo district in the optimal solution. Moreover, since the special case is not triggered for $T$, we know that $w(c_T)\le w(T)/\Delta$.
    Therefore, $\sum_{T^*\in f^{-1}(T)} w(T^*) \le (\Delta + \frac{1}{\Delta})w(T)$ as desired.
    \end{itemize}
    Since in any case the contribution of $w(T)$ in the anaylsis does not exceed $(\Delta+\frac{1}{\Delta})w(T)$, we conclude that \Cref{alg:greedy_bounded_degree} produces an $(\Delta+\frac{1}{\Delta})$-approximate solution.

    For runtime, note that the $c$-balanced district for each node can be found in time $O(n2^\Delta)$ and the sorting and greedy strategy can be implemented in time $\widetilde{O}(n\Delta 2^\Delta)$.
\end{proof}

\paragraph{Remark.} The exponential dependence on $\Delta$ in runtime can be reduced to $\poly(\Delta, 1/\epsilon)$ at the cost of an additional $(1+\varepsilon)$ term in the approximation factor. This can be achieved by using the FPTAS for complete graph (\cref{alg:complete}) for finding the largest $c$-balanced district centered at each vertex $v$ in \cref{alg:greedy_bounded_degree}. 

\paragraph{Binary Weights.} If we impose an additional constraint such that the input graph $G$ has a bounded degree $\Delta$ and each vertices have either $p_1$ or $p_2$ weight to be 1, 
we may achieve a $(\Delta+1)/2$-approximate solution.
Consider an algorithm that computes any maximum (cardinality) matching $M^*$ where for each matched edge $\{u, v\}\in M^*$ we have $p_1(u)=1$ and $p_2(v)=1$.
The algorithm simply treats each matched edge as a perfectly balanced district and returns $\cT_{\mathrm{ALG}}=M^*$. Clearly, $w(\cT_{\mathrm{ALG}}) = 2|M^*|$.

To see the approximation ratio, consider any optimal solution $\cT^*$. Each district $T\in\cT^*$ contains at least one vertex with $p_1$ weight 1 and one vertex with $p_2$ weight 1. Thus, $|M^*|\ge |\cT^*|$.
On the other hand, since each star district $T\in\cT^*$ has at most $\Delta+1$ blocks, we conclude that 
\[
w(\cT^*) \le (\Delta+1)|\cT^*| \le (\Delta+1)|M^*|=\frac{\Delta+1}{2} w(\cT_{\mathrm{ALG}}).
\]
Since we are pairing up vertices of different weight types, we are actually solving for a maximum unweighted bipartite graph matching. The state-of-the-art unweighted maximum bipartite matching can be solved in $m^{1+o(1)}$ time~\cite{ChenKLPGS23}, $\tilde{O}(m+n^{1.5})$ time~\cite{BrandLNPSS0W20}, or $n^{2+o(1)}$ time~\cite{bernstein2024maximumflowaugmentingpaths,ChuzhoyK24}.


\subsection{General graphs with binary weights and star districts}


For binary weight setting, we assume that each vertex $v$ has only one type of non-zero population/weight, i.e. either $p_1(v) = 1$ or $p_2(v) = 1$.  
We call a vertex \emph{covered} if it is part of some districting $T \in \cT$ and \emph{uncovered} otherwise.  

We use local search to find an approximate solution, with details in \cref{alg:local_search}.
We prove that this algorithm returns a $c$-approximate districting. The idea is to argue that for \textit{any} optimal district $T_i^* \in \cT^*$, at least $c$ fraction of nodes in $T_i^*$ are covered by some $T_i \in \cT$, the approximation guarantee then follows.
We do this on a case basis, by considering a fixed $T_i^*$ and for its center vertex, based on the membership in districting $\cT$ argue for the covering of its neighbor vertices.

\begin{algorithm}
\caption{Local search algorithm for star districts on graphs with binary weights}\label{alg:local_search}
  \KwData{Input graph $G = (V,E)$ with binary weights, $c\ge2$.}
\KwResult{Balanced star districts $\cT = \{T_1, \cdots , T_\ell\}$}
Initialize $\cT = \emptyset$, $k$ to be the largest integer $\le c-1$. \\
\While{$\exists (u, w) \in E$ such that $p_1(u) = 1$, $p_2(w) = 1$ and $u$ and $w$ are uncovered}
{
   Start a new district $T$, setting any one of the nodes as the center (assume $w$)\\
    Greedily assign uncovered neighbors of $w$ to $T$ while maintaining $c$-balanced property of $T$ \\
    \For{$v \in T \setminus \{w\}$}{
        \If{ $v$ has $ > k$ uncovered neighbors such that $v$ forms a $c$-balanced district with these vertices
        }{
        Remove $v$ from $T$ and start new district $T'$ with $v$ as center\tcp*{swap step}
        Assign the $> k$ uncovered neighbors to $v$.\\
        }
    }
    Add all such created $T'$ to $\cT$. \\
    If the resulting $T$ after swap steps is still balanced, add $T$ to $\cT$.
}
Output $\cT = \{T_1, \cdots , T_\ell\}$
\end{algorithm}

\begin{theorem}
\label{lem:local_search}
    The districting $\cT$ returned by \Cref{alg:local_search} on termination is a $c$-approximate solution.
\end{theorem}


\begin{proof}
    Let $\cT$ denote the set of districts returned by \Cref{alg:local_search} and $\cT^*$ be the set of optimal districts. 
    Consider $T_i^*  \in \cT^*$ in the optimal solution and let $c_i$ be its center, $\{v_1, \cdots , v_m\}$ be the children of $c_i$ in $T_i^*$ and assume wlog $p_1(c_i) = 1$.

    \paragraph{Case 1:} Suppose $c_i \notin T_i$ for all $T_i \in \cT$.
    Since we have $p_1(c_i) = 1$, the subset of vertices in $\{v_1, \cdots , v_m\}$ that are not covered by any district in $\cT$ have $p_1$ weight to be 1 (and $p_2$ weight of $0$), as otherwise the algorithm would use the edge of any of those vertices to $c_i$ to start a new district. This implies all the vertices with $p_2$ weight of 1 in $\{v_1, \cdots , v_m\}$ are covered in $\cT$ which gives us that at least $|T_i^*| / c$ vertices are covered.

    \paragraph{Case 2:} Suppose that $c_i$ is also a center of $T_i \in \cT$.
    Assume wlog that majority of vertices in $T_i$ have $p_1$ weight of 1. Since the algorithm terminates, it means that no additional vertex with $p_1$ weight 1 from $\{v_1, \cdots , v_m\}$ can be added without violating the $c$-balanced property. It also implies that vertices with $p_2$ weight 1 from $\{v_1, \cdots , v_m\}$ are covered in $\cT$ (otherwise the algorithm could just include those vertices since $T_i$ is has majority of vertices weight $p_1$ weight of 1, giving us that at least $|T_i^*| / c$ vertices are covered. The case where $T_i$ has majority vertices with $p_2$ weight 1 follows via similar argument.

    \paragraph{Case 3:} Suppose $c_i \in T_i$ but it is not the center.
    We have the following two cases which follow from the conditions on the ``swap'' step within the algorithm: 

    \begin{itemize}[itemsep=0pt]
    \item \textbf{Case (3.a):} The swap step could not create a new district with $c_i$ as the center.
    Recall we assume with loss of generality that $p_1(c_i) = 1$. Since a new district could not be created, neighbors of $c_i$ with $p_2$ weight 1 (which also include vertices from $\{v_1 , \cdots v_m\}$ in $T_i^*$) were already covered by some $T$. This in turn implies that at least $|T_i^*| / c$ vertices are covered.

    \item \textbf{Case (3.b):} Creating a new district with $c_i$ was possible but the improvement was not substantial. This case occurs when the district resulting from swapping $c_i$ as center results in at most $k$ \text{additional} vertices being covered. 
    As creating a new district was possible, we have that $c_i$ has at least one uncovered neighbor with $p_2$ weight 1. This implies at least $c-1$ vertices with either $p_i$ weight could be added to this potential district without violating the $c$-balanced property, but were not available. As $c_i$ is not a center from swap step, at most $k$ vertices from $\{v_1, \cdots , v_m\}$ are left uncovered by $\cT$. If $|T_i^*| < c$, since $c_i$ is covered, it follows that at least $1/c$ fraction of vertices are covered.
    Now, for $|T_i^*| \ge c$, $k \le c-1$ and $c \ge 2$ we have $|T_i^*| - (c-1) \ge \frac{|T_i^*|}{c}$ which gives us the desired approximation factor. 
    \end{itemize}
\end{proof}

\bibliographystyle{abbrv}
\bibliography{cover}

\appendix

\section{Omitted Proofs from Section 3}\label{omit3}

Throughout, we refer type-1 vertices as vertices that have non-zero weight $p_1$ and zero weight of $p_2$ and type-2 vertices as vertices with non-zero weight $p_2$ and zero weight of $p_1$.

\BasicHardnessResult*

\begin{proof}
We describe a polynomial reduction from the $\mathsf{NP}$-hard problem of \emph{exact set cover}: given a set $S$ of $n$ elements $x_1, x_2, \cdots, x_n$, and a family $\mathcal{S}$ of subsets $S_1, S_2, \cdots, S_m$, where $S_i \subseteq S$, is it possible to find $\mathcal{S}'\subseteq \mathcal{S}$ such that every element in $S$ is covered by exactly one set in $\mathcal{S}'$.

We construct a bipartite graph $G$ where we place one vertex for each element $x_i$ and one vertex for each set $S_j$. We connect edges between $x_i$ with $S_j$ if and only if $x_i\in S_j$. We assign $p_1(x_i)=1$ and $p_2(x_i)=0$, $\forall i$. And we assign $p_1(S_j)=0$ and $p_2(S_j)=(c-1)|S_j|$, $\forall j$.
The only $c$-balanced districts in this graph are stars, each rooted at a vertex/block $S_j$ with all neighbors in $S$. Therefore if there is an exact set cover, then there is a solution to the $c$-balanced districting problem in $G$ with total weight of $cn$. Otherwise, the best $c$-balanced districting problem in $G$ has total weight no greater than $c(n-1)$.
\end{proof}

\HardnessForPlanarGraphs*

\begin{proof}
We use a reduction from \emph{planar 1-in-3SAT}~\cite{Mulzer2008-sl}. This is a problem of $n$ Boolean variables $x_1, x_2, \dots, x_n$ and $m$ clauses each with three literals, and the incidence graph of how the variables appear in the clauses is planar. The problem answers \textsc{Yes} if there is an assignment to the Boolean variables such that each clause is satisfied by exactly one literal.

We realize the planar 1-in-3SAT instance by a planar graph. In the graph we have two types of vertices, type-1 vertices which has non-zero weight $p_1$ and zero weight of $p_2$ and type-2 vertices with non-zero weight $p_2$ and zero weight of $p_1$. We take $c=2$ thus each district must have an equal amount of the two types of weight.

\begin{figure}[htbp]
\centering
\includegraphics[width=0.95\linewidth]{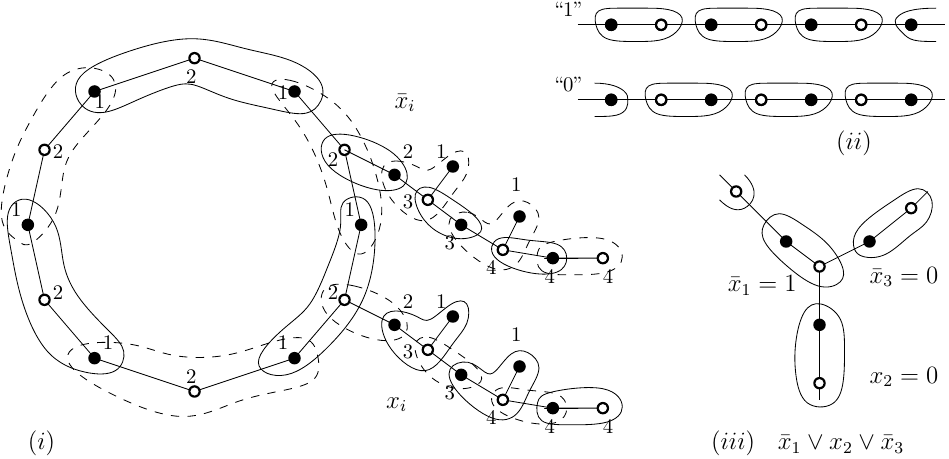} 
\caption{(i) A variable gadget for $x_i$. Solid vertices are type-2 vertices and hallow vertices are type-1 vertices. Here we use $m=3$ in this example and one clause has $x_i$ and one has $\bar{x}_i$. Thus there are two tendrils connect to clauses. If $x_i=1$ we choose solid districts; otherwise we choose dashed districts.  (ii) Implementation of an edge from a variable $x_i$ to a clause $C_j$ where $x_i$ appears. If the chain is propagating ``1", the districts on the chain starting from a variable gadget have a type-2 vertex followed by a type-1 vertex; otherwise the districts on the chain have a type-1 vertex followed by a type-2 vertex.
(iii) The clause gadget $x_1 \cup x_2 \cup \bar{x}_3$. Exactly one of $x_1, x_2, \bar{x}_3$ is assigned $1$ which means that the center type-2 vertex is covered by exactly one district.}
\label{fig:3SAT}
\end{figure}

Figure~\ref{fig:3SAT} (i) shows the gadget for a variable $x_i$. The center of this gadget is a cycle of $4m$ vertices, arranged on a cycle with alternating type-1 vertices (hallow) and type-2 vertices (solid). The type-1 vertices have weight of $p_1=2$ and type-2 vertices have weight of $p_2=1$. 
From a type-1 vertex we may grow one tendril that leads to a clause gadget, if this vertex appears in the clause. There are a total of $2m$ type-1 vertices where the tendrils of $m$ of them may lead to clauses where $x_i$ appears in the original form; and the other $m$ of them may lead to clauses where $\bar{x}_i$ appears. These are arranged in an alternating manner along the cycle. If $x_i$ appears in only $k$ clauses, there are only $k$ tendrils. Those tendrils that connect to clause gadgets carry alternating type-1 and type-2 vertices. The weights are labeled nearby the vertices in the Figure. 

The only balanced district we can find on the cycle is a star centered at a type-1 vertex $v$ with two type-2 leaves that are neighbors of $v$ on the cycle. 
Further, we could use two sets of districts (each set has $m$ balanced districts) to cover all type-2 vertices on the cycle. In the figure these two sets of districts are shown in solid sets and dashed sets respectively. They appear in alternating order along the cycle. Using either set we cover all type-2 vertices on the cycle and leave out precisely $m$ type-1 vertices uncovered. 
If we choose the solid set it means we assign $1$ to this variable. If we choose the dashed set it means we assign $0$ to this variable. Either way all type-2 vertices along the variable cycles are covered. How the type-1 vertices are covered (or not) depends on the assignment of the corresponding variable.

The choice of the districts (corresponding to the assigned Boolean value to $x_i$) is propagated to the clause set. 
The implementation of an edge connecting a variable $x_i$ to a clause $C_j$ is by a \emph{path} of alternating type-1 vertices and type-2 vertices. See Figure~\ref{fig:3SAT} (ii).
Each path starts from a type-1 vertex on the variable gadget and stops at a type-1 vertex at the center of a clause gadget. The path can be propagating a value of `1' if the covering tree takes a type-2 vertex followed by a type-1 vertex along the direction from the variable gadget to the clause gadget (solid before hallow), or `0' if swapped. When the first type-1 vertex along a tendril is covered by the variable assignment choice on the variable gadget (meaning the corresponding literal is taken as a 1), the districts on the tendril will be propagating a ``1'' to the clause.

One subtle issue of Figure~\ref{fig:3SAT} (i) is the covered weights of type-1 vertices, because the number of times a variable $x_i$ appears in the clauses can be different from that of $\bar{x}_i$. Suppose we take $x_i=1$ (i.e., the solid districts), the tendrils corresponding to $\bar{x}_i$ grab a type-1 (hallow) vertex on the cycle, while the type-1 vertices in the dashed districts on the variable gadget that are not attached to any tendrils will not be covered. Thus, the weight of type-1 vertices in the cycle and tendrils uncovered by $x_i=1$ is $2(m-k(x_i))$, with $k(x_i)$ as the number of times that $\bar{x}_i$ appears in clauses. To make sure that we always leave out the same total weight uncovered regardless of $x_i=1$ or $x_i=0$, we amend the beginning part of the tendrils such that if it is propagating a ``0'', we leave out an extra weight of two not covered. We achieve this by gradually increasing the node weight from $2$ on the cycle to $4$ with two attached vertices of  each of type-2 weight $1$. Figure~\ref{fig:3SAT} (i) shows how the two extra vertices are attached. This way, regardless of whether $x_i=1$ and $x_i=0$, the total weights not covered by the variable gadget along with the beginning parts of the tendrils are exactly $2m$.

The clause gadget is implemented by a single type-1 vertex where three tendrils (corresponding to three literals) meet. 
If $x_i$ (or $\bar{x}_i$) appears in the clause, we connect a path from a type-1 vertex covered by a solid (or dashed) district of the variable gadget $x_i$. Thus if the solid district is chosen, the path from $x_i$ propagates an assignment of $1$ to the clause and covers the center. If the dashed tree cover is chosen, the path from $x_i$ propagates an assignment of $0$ which cannot cover the center of the clause gadget. 
The center is covered by the (unique) path that is assigned ``1''. 
See Figure~\ref{fig:3SAT} (iii) for an example.

In summary, if the planar 1-in-3SAT is satisfiable, all $m$ type-2 vertices of the clause gadgets are covered by disjoint balanced districts. If the planar 1-in-3SAT cannot be satisfied, any balanced districting will leave at least some type-2 vertices not covered. This shows that the $c$-balanced districting problem in a planar graph is NP-hard.
\end{proof}

\HardnessForCompleteGraphs*

\newcommand{\SubsetSum}{\textsc{SubsetSum}\xspace}
\begin{proof}
We start from an instance of \SubsetSum and turn it into an instance of $c$-balanced districting problem. In \SubsetSum, there are $n$ numbers $x_1, x_2, \cdots, x_n$ and the goal is to ask if any subset sums up to $X=\sum_i x_i/2$. Now we create a complete graph $G$ with $n+1$ vertices including a vertex $v_0$ of type-2 weight of $X$, and type-1 weight of $0$, and $n$ vertices $v_1, \cdots v_n$ with $v_i$ holding type-1 weight of $(c-1)\cdot x_i$ and type-2 weight of $0$. 

Therefore, if the \SubsetSum answers positively, there is a star with root $v_0$ and the vertices in the \SubsetSum solution that achieves the maximum possible covered population of $c\cdot X$. On the other hand, if $c$-balanced districting problem returns a solution with covered population of $c\cdot X$, then the answer to the \SubsetSum instance is true.

Similarly, we can reduce a \SubsetSum instance to a star, where $v_0$ is the root of the tree, and $v_1, \cdots, v_n$ are leaves. The same argument follows.
\end{proof}



\HardnessForApproximation*

\begin{proof}
We take an instance of maximum independent set problem and turn it into an instance of $c$-balanced districting problem. Given a graph $G=(V, E)$, $n=|V|$ and $m=|E|$. Denote by $\Delta$ the maximum degree of $G$. For each vertex $v\in V$, we create a vertex $v'$ in graph $G'$ with type-1 weight of $(c-1)\Delta$ and type-2 weight of $0$. There are $n=|V|=|V'|$ such ``type-1 vertices''. We also have a set of ``type-2 vertices'' $V''$ with type-2 weight of $1$ and type-1 weight of $0$. Each vertex $v'\in V'$ has exactly $\Delta$ type-2 neighbors. If $u, v$ has an edge in $G$, in $G'$, $u'$ and $v'$ share one type-2 vertex, which corresponds to the edge between $u, v$ in $G$. See \Cref{fig:MIS}. If the degree of $u$ is less than $\Delta$, the corresponding vertex in $G'$ may have some dangling (degree-1) type-2 vertices. The total number of type-2 vertices is $n\Delta-m$. Thus the total number of vertices in $G'$ is $n(\Delta+1)-m$ and the number of edges in $G'$ is $n\Delta$. In order for a type-1 vertex $u'$ to be covered, all its type-2 neighbors must be used. Thus a maximum independent set $S$ in $G$ means we can cover all corresponding vertices of $S$ in $G'$ as well as all their type-2 neighbors, leading to a total coverage population of $|S|\Delta$. Similarly, if we can find a $c$-balanced districting problem in $G'$, the type-1 vertices that are covered in $c$-balanced districts cannot share any common type-2 neighbors, and therefore the corresponding vertices in $G$ must be independent. This reduction works when the district must be a star. 

\begin{figure}[htbp]
\centering
\includegraphics[width=0.45\linewidth]{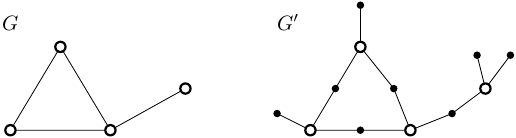} 
\caption{Graph $G$ and $G'$. $\Delta=3$ in this instance. The hallow vertices of $G'$ have type-1 weight of $3(c-1)$ and the solid vertices of $G'$ have type-2 weight of $1$.}
\label{fig:MIS}
\end{figure}

This reduction shows hardness of approximation, as an $\alpha$-approximation for maximum independent set means an $\alpha$-approximation for $c$-balanced districting problem, for any $c$. 
The maximum independent set cannot be approximated by a factor of $n^{1-\delta}$ for any constant $\delta > 0$ on general graph~\cite{bazganpolyAPX05, Zuckerman06}.
If we have an approximation algorithm for the districting problem with approximation factor $O(N^{1/2-\delta})$ with $N$ as the number of vertices in the districting graph $G'$, by the reduction $N=O(n\Delta)$ and this gives an $O((n\Delta)^{1/2-\delta})=O(n^{1-2\delta})$ for the maximum independent set problem on $G$, since $\Delta<n$, which is impossible unless $\mathsf{P}=\mathsf{NP}$.

As the maximum degree in both $G$ and $G'$ is $\Delta$, approximating the balanced districting problem in $G'$ with some factor depending on $\Delta$ gives the same approximation factor for the maximum independent set in $G$. 
For bounded degree graphs, the maximum independent set has a constant approximation~\cite{Halldorsson1994-as}, but is APX-complete~\cite{Chlebik2003-cu} and cannot expect an approximation ratio better than $\Delta/2^{O(\sqrt{\Delta})}$ unless $\mathsf{P}=\mathsf{NP}$ \cite{Samorodnitsky00PCP}. Further, assuming the Unique Games Conjecture, one cannot approximate the maximum independent set problem within a factor of $O(\Delta / \log^2 \Delta)$ \cite{AustrinKS11}. These (conditional) hardness results extend to balanced star districting problem on graphs with maximum degree $\Delta$.
This finishes the argument. 
\end{proof}

\section[A Whack-A-Mole-Algorithm For c-Balanced Star Districting]{A Whack-A-Mole Algorithm For $c$-Balanced Star Districting}\label{sec:whack-a-mole}

\subsection{The MWU Framework}

We first recap the classical multiplicative weight update framework.

Consider the following online process.
There are $n\in\mathbb{N}$ experts and there is a weight vector $y^{(1)}=\frac{1}{n}\cdot \mathbf{1}^n$. Let $\rho\in\mathbb{R}_{> 0}$ be a positive value that represents the bound of an expert's gain in a single round.
Let $T := \lceil\rho \epsilon^{-2}\ln n\rceil $ be the target number of rounds.
For each round $t=1, 2, \ldots, T$, the process receives a \emph{reward} vector $g^{(t)}=[g_1^{(t)}, \ldots, g_n^{(t)}]\in [-\rho, \rho]^n$.
Then, the process computes a temporary vector $\hat{y}^{(t+1)}_i := (1+\epsilon\cdot g_i^{(t)} / \rho) \cdot y^{(t)}_i$ for all $i=1, 2, \ldots, n$, and then normalizes it, obtaining $y^{(t+1)} := \hat{y}^{(t+1)} / \| \hat{y}^{(t+1)} \|$.
After the execution of all $T$ rounds, we have the following guarantee (from \cite[Corollary 2.6]{v008a006}). 

\begin{theorem}\label{thm:mwu}
Fix any $i$.
Let $G_i := \sum_{t=1}^{T} g_i^{(t)}$ be the sum of all rewards of the $i$-th expert throughout the process. 
Let $G_{\mathrm{ALG}} := \sum_{t=1}^T \langle y^{(t)}, g_i^{(t)}\rangle $ be the expected total gain, if in each round $t$ the prediction algorithm randomly picks an expert according to the distribution $y^{(t)}$, and follows that expert's decision.
Then,
\[
\frac{G_i}{T} \le (1+\epsilon)\frac{G_{\mathrm{ALG}}}{T} + \epsilon.
\]
\end{theorem}

\subsection{The Whack-A-Mole Algorithm}

Let us focus on the dual linear program.
Without loss of generality we assume that each $c$-balanced district $S$ has positive population, namely, $w(S)\ge 1$.
Let $\mu$ be the binary searched value to the optimal objective value of our linear programs.
It is straightforward to check that $\mu\in [1, w(G)]$, where $w(G)$ is the total weight of all vertices in the graph.
Given that the re-weighted linear program where each coefficient of the form $1/w(S)$ is replaced with $\mu/w(S)$,
the goal of the whack-a-mole algorithm now is to either find a feasible primal solution of value $\ge 1-2\epsilon$, or a feasible dual solution of value $\le 1+\epsilon$. This ensures an $(1+3\epsilon)$-approximation in the end. If we aim for an $(1+\epsilon')$-approximation with a specific $\epsilon' > 0$, then we can always initialize $\epsilon := \epsilon'/3$.  

We first describe the slow version of the whack-a-mole algorithm.
This algorithm simulates the prediction game round by round with a total of $T:=\mu\epsilon^{-2}\ln n$ rounds.
We set $\rho = \mu$. The algorithm initializes $x_S=0$ for all $S\in\mathcal{S}$ and $y_v=1/n$ for all $v\in V$.
In each round $t$, the whack-a-mole algorithm invokes the separation oracle and finds a weakly violating $c$-balanced district $S$ with 
\begin{equation}\label{eq:whack-a-mole-violation}
\frac{\mu}{w(S)} \cdot \langle y^{(t)}, \mathbf{1}_S\rangle < 1-\frac{\epsilon}{2}.
\end{equation}
The whack-a-mole algorithm then sets a reward vector $g^{(t)} := \frac{\mu}{w(S)}\mathbf{1}_S$. Thus, $\langle y^{(t)}, g^{(t)} \rangle < 1-\epsilon/2$. Meanwhile, the algorithm also increases the variable $x_S$ by $1$.
This process repeats for round $t=1, 2, \ldots,$ until either (1) $t=T$, or (2) there is no strongly violating $c$-balanced district anymore.

In case (1), the algorithm returns $x'_S \gets (1-2\epsilon)x_S/T$ as a feasible primal solution. Note that now $\sum_{S} x'_S = 1-2\epsilon$. This indicates that $\mathit{OPT}\ge \mu$.
In case (2), the algorithm returns a feasible dual solution $y'_v\gets (1-\epsilon)^{-1}\cdot y_v^{(t)}$, where $t$ refers to the last round of the WMU step. Note that now $\sum_v y'_v = (1-\epsilon)^{-1}$. This indicates that $\mathit{OPT}\le (1+2\epsilon)\mu$ whenever $\epsilon \le 1/2$. 
Now, we show that the algorithm does halt with a feasible (either primal or dual) solution.

\begin{lemma}
If the algorithm ends in case (1), then $\{x'_S\}$ is a feasible primal solution.
\end{lemma}

\begin{proof}
In case (1), $T$ rounds of the MWU steps are performed. Hence, we are able to utilize \Cref{thm:mwu} to validate the primal constraints.
For every $v\in V$, the constraint in \textsc{(Primal)} states that
\begin{align*}
\sum_{S\in\mathcal{S}: S\ni v} \frac{\mu}{w(S)} x'_S &= (1-2\epsilon)\cdot\frac{1}{T} \sum_{S\in\mathcal{S}: S\ni v} \frac{\mu}{w(S)} x_S \\
&= (1-2\epsilon)\cdot\frac{1}{T} \sum_{S\in\mathcal{S}: S\ni v} \frac{\mu}{w(S)}\cdot (\text{\# of rounds where $S$ is chosen})\\
&=(1-2\epsilon)\cdot\frac{1}{T} \sum_{t=1}^T \frac{\mu}{w(S^{(t)})}\mathbb{I}[g_v^{(t)} \neq 0] \tag{$\mathbb{I}[\cdot]$ is the indicator function}\\
&=(1-2\epsilon)\cdot\frac{1}{T} \sum_{t=1}^T \frac{\mu}{w(S^{(t)})} \cdot g_v^{(t)}\cdot \frac{w(S^{(t)})}{\mu} \tag{by definition of $g_v^{(t)}$}\\
&=(1-2\epsilon)\cdot\frac{1}{T} \sum_{t=1}^T g_v^{(t)} = (1-2\epsilon) \frac{G_v}{T} \\
&\le (1-2\epsilon)\cdot  \left(\epsilon + (1+\epsilon)\frac{G_{\mathrm{ALG}}}{T}\right) \tag{By \Cref{thm:mwu}}\\
&=(1-2\epsilon)\cdot \left(\epsilon + \frac{1+\epsilon}{T}\sum_{t=1}^T \langle g^{(t)}, y^{(t)}\rangle \right) \\
&\le (1-2\epsilon)\cdot \left(\epsilon + \frac{1+\epsilon}{T}\sum_{t=1}^T (1-\frac{\epsilon}{2})\right) \tag{By \Cref{eq:whack-a-mole-violation}}\\
&\le (1-2\epsilon)\cdot(1+\frac{3}{2}\epsilon) \le 1
\end{align*}
\end{proof}

\begin{lemma}
If the algorithm ends in case (2), then $\{y'_v\}$ is a feasible dual solution.
\end{lemma}

\begin{proof}
In case (2), the separation oracle certifies that there is no strongly violating constraints. Thus, from the definition of a strongly violation constraint (\Cref{eq:violate}), we know that for all $c$-balanced star district $S$, we have 
$\sum_{v\in S} y_v^{(t)}\ge (1-\epsilon) \cdot w(S)$, which implies that $\sum_{v\in S} y'_v \ge w(S)$. Thus, $\{y'_v\}$ is a feasible dual solution.
\end{proof}

\subsection{Speeding Up The Algorithm}

The previous algorithm requires $T$ rounds of WMU steps.
Unfortunately, $T=\lceil\mu\epsilon^{-2}\ln n\rceil$ and $\mu$ can be as large as $w(G)$.
A straightforward implementation requires $\Omega(w(G))$ time, even if we are provided with a polynomial time separation oracle.
From the description of our algorithm, we observe that without normalization, the sum $\sum_{v} y_v$ is non-decreasing. 
The main idea to speed up the MWU algorithm (provided in \cite{BKS23-lp}), is to perform a lazy normalization of the dual variables: the algorithm normalizes $\{y_v\}$ whenever its sum exceeds $1+\epsilon$.
The MWU rounds between two normalizations is then called a \emph{phase}.
Within the same phase, the algorithm is allowed to perform multiple MWU steps on the same violating constraint $S$ \emph{at once} until one of the following three situations happen: (1) this constraint becomes satisfied, (2) the total number of steps now reaches $T$, or (3) the sum of dual variables exceeds $1+\epsilon$ so the algorithm enters a new phase.
This subroutine is called \textsc{Enforce} in \cite{BKS23-lp}. 
We are now able to give a polynomial upper bound to the number of phases in total.

\begin{lemma}\label{thm:mwu-phases}
The number of phases is at most $O(\epsilon^{-2}\log n)$.
\end{lemma}

\begin{proof}
There are $T:= \lceil\mu  \epsilon^{-2} \ln n\rceil$ MWU rounds. 
Let $\hat{y}_v^{(t)}$ be the  unnormalized dual variables and let $\|\hat{y}^{(t)}\|=\sum_v\hat{y}_v^{(t)}$ be the sum at the beginning of the round $t$.
Initially, the sum of all dual variables is $\sum_v \hat{y}_v^{(1)}=1$.
Within each round $t$, a (weaker) violating constraint $S$ is chosen such that
\[
\sum_{v\in S} \hat{y}_v^{(t)} \le (1-\epsilon/2)\cdot \frac{w(S)}{\mu}\cdot \|\hat{y}^{(t)}\|.
\] 
A MWU step multiplies each $\hat{y}_v^{(t)}$ by $(1+\epsilon \cdot g_v^{(t)} /\mu)$ for $v\in S$. Notice that $g_v^{(t)} := \mu/w(S)$ for all $v\in S$.
This implies that the absolute increase of the unnormalized dual variables can be bounded by 
\begin{align*}
\|\hat{y}^{(t+1)}\| - \|\hat{y}^{(t)}\| &\le \sum_{v\in S}
\frac{\epsilon}{\mu}\cdot g_v^{(t)}\cdot \hat{y}_v^{(t)} \\
&=  \frac{\epsilon}{\mu} \cdot \frac{\mu}{w(S)} \cdot \sum_{v\in S} \hat{y}_v^{(t)}\\
& \le \frac{\epsilon}{\mu} \cdot \frac{\mu}{w(S)} \cdot (1-\epsilon/2) \cdot \frac{w(S)}{\mu} \cdot \|\hat{y}_v^{(t)}\|\\
&\le \frac{\epsilon}{\mu} \|\hat{y}_v^{(t)}\|.
\end{align*}
Thus, each MWU step increases the total sum by a factor of at most $1+\epsilon/\mu$. Therefore,
\begin{align*}
\|\hat{y}^{(T)}\| &\le (1+\epsilon/\mu)^T 
= (1+\epsilon/\mu)^{\lceil\mu\epsilon^{-2}\ln n\rceil}\\
&\le e^{(\ln n)/\epsilon+1} = e\cdot n^{1/\epsilon}.
\end{align*}
Notice that the algorithm is triggered to start a new phase whenever the current normalized sum of the dual variables $\sum_v y_v \ge 1+\epsilon$. Hence, the total number of phases can be bounded by
\begin{align*}
\log_{1+\epsilon} \left(e\cdot n^{1/\epsilon}\right) &=O(\epsilon^{-2}\ln n),
\end{align*}
as desired.
\end{proof}

\subsection[Separation Oracle: Proof of Lemma 4.1]{Separation Oracle: Proof of \Cref{lem:separation-oracle}}\label{sec:appendix-solve-separation-oracle}
\LemmaSeparationOracle*

\begin{proof}
We generalize \Cref{alg:complete} to accommodate dual variables.
In particular, we will maintain a candidate list $L$ with the following property:
For any $c$-balanced district $S$, there exists a district $T\in L$ such that (1) $(p_1(S), p_2(S))$ is an $(\epsilon/10)$-approximate of $(p_1(T), p_2(T))$, and (2) $\sum_{v\in S} w'(v) / \sum_{v\in T} w'(v)\in [e^{-\epsilon/10}, e^{\epsilon/10}]$. Recall that from the proof of \Cref{thm:fptas-complete} we defined that a pair of numbers $(q_1, q_2)$ is an $\epsilon$-approximate to $(q'_1, q'_2)$ if both $q_1/q'_1, q_2/q'_2\in [e^{-\epsilon}, e^\epsilon]$.
The above property suggests that we add a third dimension for $y'_v$ to the list, and modify the trimming algorithm slightly --- we will not trim the solution if their $\sum_{v\in S} y'_v$ values are too far from each other.

We now prove that the final list contains a $c$-balanced district that is a weakly violating constraint.
Consider the population of commodities $(p_1(S_\mathrm{max}), p_2(S_\mathrm{max}))$ of $S_\mathrm{max}$. Without loss of generality, assume that $p_1(S_\mathrm{max})\ge p_2(S_\mathrm{max})$. 
Then, by the property we stated above, there exists a district $S\in L$ that satisfies:
\begin{align*}
(c-1)p_2(S) - p_1(S) & \ge (c-1)p_2(S_\mathrm{max}) - p_1(S_\mathrm{max}) \tag{maintained by \Cref{alg:complete}}\\
&\ge 0, \text{ and}\\
(c-1)p_1(S) - p_2(S) & \ge (1-\epsilon/10) (c-1)p_1(S_\mathrm{max}) - (1+\epsilon/10) p_2(S_\mathrm{max})\\
&\ge \left((1-\epsilon/10)(c-2) - (\epsilon/5)\right) p_2(S_\mathrm{max}) \\
&\ge 0 \tag{whenever $\epsilon\le \frac{c-2}{c}$}
\end{align*}
The above inequality shows that $S$ is indeed $c$-balanced. Furthermore, we have $\sum_{v\in S} y'_v/\sum_{v\in S_{\mathrm{max}}} y'_v\in [e^{-\epsilon/10}, e^{\epsilon/10}]$. Thus, we are able to show that $S$ is a weakly violating constraint:
\begin{align*}
\sum_{v\in S} y'_v &\le e^{\epsilon/10}\cdot \sum_{v\in S_{\mathrm{max}}} y'_v\\
&\le e^{\epsilon/10}\cdot (1-\epsilon) \cdot  w'(S_{\mathrm{max}}) \tag{$S_\mathrm{max}$ is strongly violating}\\
&\le e^{\epsilon/10}\cdot (1-\epsilon) \cdot e^{\epsilon/10} \cdot w'(S) \\
&\le e^{\epsilon/5} \cdot (1-\epsilon) \cdot w'(S) \le (1-\epsilon/2) \cdot w'(S). \tag{$\epsilon > 0$}\\
\intertext{On the other hand, we have}
w'(S) &\ge e^{-\epsilon/10}\cdot w'(S_\mathrm{max}) 
\ge \frac{1}{2} w'(S_\mathrm{max}),
\end{align*}
certifying that the output $S$ satisfies all the constraints from the lemma statement.

Let us now analyze the runtime of the algorithm. It suffices to analyze the number of scales at the new dimension.
Since each value $y'_v$ is at least $n^{-(1+1/\epsilon)}$ and is at most $1+\epsilon$, the number of scales in the third dimension can be bounded by 
\[
\log_{e^{\epsilon/10}} \frac{1+\epsilon}{n^{-(1+1/\epsilon)}} = \frac{\ln (1+\epsilon) + (1+1/\epsilon)\ln n}{\epsilon/10} = O(\epsilon^{-2}\ln n).
\]
Together with the analysis in \Cref{thm:fptas-complete}, the runtime of this generalized algorithm for the complete graph, including maintaining a solution, is
$O(\epsilon^{-6} n^6 (\log n)(\log w(G))^4)$.
\end{proof}

\subsection[Piecing Everything Together: Proof of Theorem 6]{Piecing Everything Together: Proof of \Cref{theorem:star-lp-main}}\label{sec:appendix-proof-of-thm-star-lp-main}


\paragraph{Correctness.}
The correctness comes from the whack-a-mole framework~\cite{BKS23-lp}, the implementation to the separation oracle (\Cref{lem:separation-oracle}), and the potential method analyzed in \Cref{sec:star-district-algorithm-overview}.

\paragraph{Runtime Analysis.}
We have $O(\epsilon^{-1}\log w(G))$ for binary search up to an $(1+\epsilon)$-approximation.
For each guess $\mu$, we run the whack-a-mole algorithm which has $O(\epsilon^{-2}\log n)$ phases and within each phase there are at most $O(n^2/\epsilon)$ calls to the separation oracle by \Cref{eq:7}. Each call to the separation oracle takes $O(\epsilon^{-6}n^{6}(\log n)(\log w(G))^4))$ by \Cref{lem:separation-oracle}.
Therefore, the entire algorithm runs in time
\[
O(\epsilon^{-10} n^{8} (\log n)^2 (\log w(G))^5) = \text{poly}(n, 1/\epsilon, \log w(G)).
\]

\paragraph{Number of Non-Zero Primal Variables.} It is probably worth to note that the number of non-zero terms does not involve any $\log w(G)$ terms. In particular, the algorithm uses the last returned feasible primal solution as the approximated solution. 
The number of non-zero terms can then be bounded by the number of phases multiplied by the number of separation oracle calls per phase, which is at most $O(\epsilon^{-3} n^2\log n)$.

\subsection{Proof of Randomized Rounding}
\begin{lemma}\label{lem:randomized-rounding}
The randomized rounding algorithm from the fractional LP solution produces an expected weight for output districting $I$ as
\[
\mathbf{E}[w(I)] \ge \sum_{S\in \mathcal{S}_{\mathrm{LP}}} w(S) \frac{x_S}{\tau} - \sum_{A, B\in \mathcal{S}_{\mathrm{LP}}:\ A\cap B\neq \emptyset} \min(w(A), w(B)) \frac{x_Ax_B}{\tau^2}\ .
\]
\end{lemma}
\begin{proof}
Consider a district $S$ with non-zero value $x_S$, it is selected into $I$ only if two events happen: (1) the coin flip with probability $x_S/\tau$ turns out to be true; and (2) all the districts with value at least $x_S$ are not included in $I$ -- their coin flips are false. The probability of both events happening is 
\[
\frac{x_S}{\tau} \cdot \prod_{A\in \mathcal{S}_{\mathrm{LP}}: A\cap S\neq \emptyset, w(A) \ge w(S)} \left(1-\frac{x_A}{\tau}\right) \geq \frac{x_S}{\tau} \cdot \bigg(1-\sum_{A\in \mathcal{S}_{\mathrm{LP}}: A\cap S\neq \emptyset, w(A)\ge w(S)}\frac{x_A}{\tau}\bigg)
\]
Now, by linearity of expectation, we have 
\begin{align*}
\mathbf{E}[w(I)] &\ge \sum_{S\in \mathcal{S}_{\mathrm{LP}}} w(S) \frac{x_S}{\tau} \cdot \bigg(1-\sum_{A\in \mathcal{S}_{\mathrm{LP}}: A\cap S\neq \emptyset, w(A)\ge w(S)}\frac{x_A}{\tau}\bigg)
\\
&=\sum_{S\in \mathcal{S}_{\mathrm{LP}}} w(S) \cdot \frac{x_S}{\tau} -\sum_{S, A\in \mathcal{S}_{\mathrm{LP}}: A\cap S\neq \emptyset, w(A)\ge w(S)} w(S)\cdot \frac{x_Sx_A}{\tau^2} \\
&= \sum_{S\in \mathcal{S}_{\mathrm{LP}}} w(S) \cdot \frac{x_S}{\tau} - \sum_{A, B\in \mathcal{S}_{\mathrm{LP}}:\ A\cap B\neq \emptyset} \min(w(A), w(B)) \cdot \frac{x_Ax_B}{\tau^2}\ .
\end{align*}
\end{proof}

For any $\delta\ge 0$, let $\mathcal{S}_{\ge\delta}$ be the set of all districts $S\subseteq \mathcal{S}_{\mathrm{LP}}$ whose weight is at least $\delta$.
The following lemma connects the unweighted correlation between overlapped districts and the expected approximation ratio to the randomized rounding algorithm.

\begin{lemma}\label{lem:rounding-tau}
Let $\tau\in\mathbb{R}_{>0}$ be a fixed value. 
Suppose that for all $\delta > 0$, 
\begin{equation}\label{eq:rounding-constraint}
\sum_{A, B\in \mathcal{S}_{\ge\delta}: A\cap B\neq\emptyset} x_Ax_B \le (\tau/2)\cdot \sum_{S\in \mathcal{S}_{\ge\delta}} x_S,
\end{equation} then $\mathbf{E}[w(I)]\ge \frac{1}{2\tau} \sum_{S} w(S)\cdot x_S$.
\end{lemma}

\begin{proof}
We first sort all districts in $\mathcal{S}_{\mathrm{LP}}$  in the non-increasing order of weights.
Let $S_1, S_2, \ldots, S_t$ be such a list.
For each district $S_i$, its weight $w(S_i)$ can be written as 
\[
w(S_i) = \sum_{j=i}^{t} (w(S_j) - w(S_{j+1})).
\]
Here for convenience we define $w(S_{t+1})=0$.
Using the above expression, we are able to establish that
\begin{align*}
&\sum_{A, B\in \mathcal{S}_{\mathrm{LP}}:\ A\cap B\neq \emptyset}\!\!\!\!\!\!\!\! \min(w(A), w(B)) \cdot \frac{x_Ax_B}{\tau^2}\\
&= \sum_{i=1}^t \sum_{\ell=1}^{i-1} \mathbb{I}[S_\ell\cap S_i\neq \emptyset] \cdot w(S_i)\cdot \frac{x_{S_i}x_{S_\ell}}{\tau^2}\\
&= \sum_{i=1}^t \sum_{\ell=1}^{i-1} \mathbb{I}[S_\ell\cap S_i\neq \emptyset] \cdot \left(\sum_{j=i}^t \left(w(S_j)-w(S_{j+1})\right)\right)\cdot \frac{x_{S_i}x_{S_\ell}}{\tau^2}\\
&= \sum_{j=1}^t \bigg(\!w(S_j)-w(S_{j+1})\!\bigg)\cdot 
 \left(\sum_{i=1}^j \sum_{\ell=1}^{i-1} \mathbb{I}[S_\ell\cap S_i\neq \emptyset] \cdot \frac{x_{S_i}x_{S_\ell}}{\tau^2} \right) \\
&=\sum_{j=1}^t \bigg(\!w(S_j)-w(S_{j+1})\!\bigg)\cdot 
 \left(
 \frac{1}{\tau^2}\cdot 
 \sum_{A, B\in \mathcal{S}_{\ge w(S_j)}:\ A\cap B\neq\emptyset} x_Ax_B
 \right)\\
&\le \sum_{j=1}^t \bigg(\!w(S_j)-w(S_{j+1})\!\bigg)\cdot 
 \left(
 \frac{1}{\tau^2}\cdot \frac{\tau}{2}\cdot \sum_{S\in \mathcal{S}_{\ge w(S_j)}} x_S
 \right) \tag{by \eqref{eq:rounding-constraint}}\\
&=\frac{1}{2\tau} \sum_{j=1}^t \sum_{i=1}^j \bigg(\!w(S_j)-w(S_{j+1})\!\bigg) \cdot x_{S_i} \\
&= \frac{1}{2\tau} \sum_{i=1}^t x_{S_i} \cdot \sum_{j=i}^t \bigg(\!w(S_j)-w(S_{j+1})\!\bigg)  \\
& = \frac{1}{2\tau}  \sum_{i=1}^t w(S_i) \cdot x_{S_i}
\end{align*}
Finally, we have
\begin{align*}
\mathbf{E}[w(I)] &\ge 
\sum_{S\in \mathcal{S}_{\mathrm{LP}}} w(S) \frac{x_S}{\tau} - \sum_{A, B\in \mathcal{S}_{\mathrm{LP}}:\ A\cap B\neq \emptyset} \min(w(A), w(B))\frac{x_Ax_B}{\tau^2} \tag{by \Cref{lem:randomized-rounding}}\\
&\ge \frac{1}{\tau}\sum_{S\in \mathcal{S}_{\mathrm{LP}}} w(S) x_S - \frac{1}{2\tau}\sum_{S\in\mathcal{S}_{\mathrm{LP}}} w(S)x_S\\
&= \frac{1}{2\tau}\sum_{S\in\mathcal{S}_{\mathrm{LP}}} w(S)x_S
\end{align*}
as desired.
\end{proof}

\subsection{Outerplanar Graphs}\label{sec:appendix-outerplanar-graphs}

In this section we devote in proving the following lemma.

\OuterPlanarGraphRoundingGap*

Let $G$ be a connected outerplanar graph.
We begin with the classical result~\cite{Diestel2017-tk}.

\begin{fact}\label{fact:outerplanar}
A graph $G$ is outerplanar if and only if $G$ does not have $K_{2, 3}$ and $K_4$ as minor.
\end{fact}

Let $r\in V$ be any vertex in $G$. 
By considering shortest distances of vertices on $G$ from $r$, we are able to partition $V$ into layers $V=L_0\cup L_1\cup \cdots \cup L_d$, where $L_i=\{v\in V\ |\ \mathrm{dist}_G(v, r)=i\}$. Notice that algorithmically we can construct the sets $\{L_i\}$ using a single BFS, but here we do not need them.
But with the BFS tree in mind, we obtain the following fact.

\begin{fact}\label{fact:claw}
Consider any three vertices from the same layer $x, y, z\in L_i$.
Then $\{r, x, y, z\}$ form a claw minor ($K_{1, 3}$-minor) in $G$.
\end{fact}

The following lemma is crucial for proving \Cref{lem:outerplanar}.

\begin{lemma}\label{lem:outerplanar-bfslayer}
For each layer $L_i$, $i\geq 1$, the induced subgraph $G[L_i]$ is a collection of paths.
\end{lemma}

\begin{proof}
We establish the proof by contradiction.
Let $H=G[L_i]$. If $H$ is not a collection of paths, there must exists either (1) a vertex $v$ with degree at least 3 in $H$, or (2) a cycle in $H$. Let $v_1, v_2, \ldots$ be the clockwise ordering of vertices from $H$ in an outerplanar embedding of $G$.  In case (1), the vertex $v$ connects to at least three other vertices $x, y, z\in L_i$. In this case, by \Cref{fact:claw} we know that $\{r, v, x, y, z\}$ is a $K_{2, 3}$ minor, which is impossible since $G$ is outerplanar.
Similarly, in case (2), a cycle has at least three vertices. Consider any three vertices $x, y, z$ within the same cycle of $H$. This implies that $\{r, x, y, z\}$ is a $K_4$ minor of $G$, a contradiction.
\end{proof}

\begin{lemma}\label{lem:outerplanar-5-hop-nbrs}
Let $G_i = G[L_i\cup L_{i+1}\cup \cdots \cup L_d]$ and let $k\in\mathbb{Z}_{\ge 0}$ be a non-negative integer.
For each vertex $v\in L_i$, the set $N_{G_i, k}(v) := \{x\in L_i \ |\ \mathrm{dist}_{G_i}(v, x) \le k\}$ has at most $2^k+1$ vertices.
\end{lemma}
\begin{proof}
We prove by induction on $k$. When $k=0$ the statement is trivially true. Now we take $k=1$. For any $i$ and any vertex $v\in L_i$,
denote by 
$X$ the neighbors of $v$ in $G_i$. 
We claim $|X|\le 2$. Indeed, if $|X| \ge 3$ then there must be a claw minor from $v$ to $\{x, y, z\}\subseteq X$ in $G_i$. This implies that $\{v, x, y, z, r\}$ is a $K_{2, 3}$ minor in $G$, a contradiction to \Cref{fact:outerplanar}. This immediately implies that $|N_{G_i, k}(v)|\leq 2$ as well.

Suppose now that the statement is correct up to $k-1\geq 1$.
Consider any vertex $v\in L_i$. 
By \Cref{lem:outerplanar-bfslayer}, $G[L_i]$ is a collection of paths. Let $X$ be the neighbor of $v$ in $L_i$.
Consider any path $P$ starting from $v$ and uses vertices in $L_{i+1}\cup L_{i+2}\cup \cdots\cup L_d$ as internal vertices and ending at any vertex on $L_i\setminus X$. We note that these paths will have lengths at least 2.
Let $Y\subseteq L_i$ be all possible ending vertices of such path $P$.
We claim that $|X|+|Y|\le 2$. Indeed, if otherwise $|X|+|Y|\ge 3$, then $\{v\}\cup X\cup Y$ is clearly containing a $K_{1, 3}$ minor.
We now observe that 
\begin{equation}\label{eq:outerplanar-nbrs}
N_{G_i, k}(v) \subseteq \left(\bigcup_{x\in X}N_{G_i, k-1}(x)\right) \cup \left(\bigcup_{y\in Y}N_{G_i, k-2}(y)\right).
\end{equation}
By induction hypothesis, for all $0\le k'<k$ and for all $x\in X\cup Y$, $|N_{G_i, k'}(x)|\leq 2^{k'}+1$. To apply the induction, it suffices to consider several cases:
\begin{itemize}[itemsep=0pt]
\item If $|X\cup Y|=0$, then $|N_{G_i, k}(v)|=1$. 
\item If $|X\cup Y|=1$, then $|N_{G_i, k}(v)|\le 2^{k-1}+1<2^k+1$.
\item If $|X|=2$ (which implies $|Y|=0$), we notice that $k-1\ge 1$ and thus $v$ is included in every set $N_{G_i, k-1}(x)$ with $x\in X$.
Then, $|N_{G_i, k}(v)|\le 2(2^{k-1}+1)-1=2^k+1$. Here we need to subtract $1$ as $v$ is double counted.
\item If $|X\cup Y|=2$ and $|Y|\ge 1$, then using \eqref{eq:outerplanar-nbrs}, we have $|N_{G_i, k}(v)| \le (2^{k-1}+1) + (2^{k-2}+1) \le 2^{k}+1$.
\end{itemize}
\end{proof}


\begin{lemma}\label{lem:two-nbr-prev-level}
For any $i$, $L_i$ can be partitioned into $O(1)$ sets $S_{i, 1}, S_{i, 2}, \ldots, S_{i, \ell}$ such that $S_{i, j}$ is a 5-hop independent set in the subgraph $G_i=G[L_i\cup L_{i+1}\cup \cdots \cup L_d]$. Specifically, $L_i$ is $(17, 5)$-scattered in $G_i$.
\end{lemma}

\begin{proof}
We consider the following greedy coloring algorithm on all vertices of $L_i$ in any order: for each vertex $v\in L_i$, assign a color to $v$ that is different from the color of any colored vertices in $N_{G_i, 4}(v)$. 
By \Cref{lem:outerplanar-5-hop-nbrs}, $2^4+1=17$ colors suffices for coloring all vertices in $L_i$.
Observe that any two vertices of the same color are at least 5-hops away from each other.
Hence, we are able to partition $L_i$ into at most $17=O(1)$ sets, where each set $S_{i, j}$, comprised of all vertices of color $j$, is a $5$-hop independent set in $G_i$.
\end{proof}

Before we give the proof of \Cref{lem:outerplanar}, we introduce a helpful observation for the analysis.

\begin{lemma}\label{lem:outerplanar-atmost-two-parents-bfs}
For any vertex $v\in L_{i+1}$, there are at most two neighbors of $v$ in $L_i$ in $G$.
\end{lemma}

\begin{proof}
If $v\in L_{i+1}$ has at least three neighbors $x, y, z\in L_i$, we know that $\{r, x, y, z, v\}$ form a $K_{2, 3}$ minor, which is a contradiction.
\end{proof}

Now we are ready to finish the proof of \Cref{lem:outerplanar}.

\begin{proof}[Proof of \Cref{lem:outerplanar}]
The proof will be similar to the proof of \Cref{lem:divide-and-conquer-analysis}, but without the recursion. 
We first apply the following charging argument, where we charge the cost $x_Ax_B$ of each pair of intersecting districts $A\cap B\neq \emptyset$ to any vertex $v\in (A\cap B)\cup \{c_A, c_B\}$ with the smallest possible distance to the chosen root $r$ on $G$. Recall that $c_A$ and $c_B$ are the center vertices of $A$ and $B$ respectively.

Consider a vertex $v$ in the $i$-th BFS layer $L_i$.
Let $Q(v)$ be the collection of district pairs charged to $v$.
By the way we charge the pairs, we know that the total cost for $v$ can be bounded by:
\begin{equation}\label{eq:outer-planar-analysis}
\mathrm{cost}(v) := \sum_{(A, B)\in Q(v)} x_Ax_B \le 2 \cdot \sum_{S:\ c_S\in N_{G_i, 2}(v)} x_S.
\end{equation}

Now, it suffices to bound $\sum_{v\in V} \mathrm{cost}(v)$.
Consider a district $S$ with the center vertex $c_S\in L_i$.
By \eqref{eq:outer-planar-analysis}, $x_S$ is never involved in any $\mathrm{cost}(v)$ with $v\in L_{i+1}\cup L_{i+2}\cup \cdots \cup L_d$.
We observe that by \Cref{lem:outerplanar-atmost-two-parents-bfs}, $x_S$ is counted by at most $4$ times in level $i-2$, at most $2$ times in level $i-1$.
Finally, we would like to analyze the number of occurrences of $x_S$ in $\mathrm{cost}(v)$ for $v$ in level $i$. By \Cref{lem:two-nbr-prev-level}, $L_i$ can be partitioned into at most 17 5-hop independent sets $S_{i, 1}, S_{i, 2}, \ldots, S_{i, 17}$. Since each set $S_{i, j}$ is a 5-hop independent set, $x_S$ can only be involved in at most one vertex.
This implies that $x_S$ is added by at most $17$ times in level $i$. Therefore, we conclude that $\sum_{v\in V}\mathrm{cost}(v) \le 2\cdot (17+4+2) \sum_S x_S = O(1)\cdot \sum_S x_S$ as desired.
\end{proof}

\subsection{Trees}\label{sec:lp-trees}

On trees, we are able to bound the product terms for overlapping districts with just the sum of all variables.

\begin{lemma}\label{lem:lp-trees}
Let $G$ be a tree. Suppose that $\{x_S\}$ are primal variables obtained by \Cref{theorem:star-main}. Then, the sum of products of primal variables among all (unordered) pairs of overlapping districts $\{A, B\}$ satisfies $\sum_{A\cap B\neq\emptyset} x_Ax_B\le \sum x_S$.
\end{lemma}

\begin{proof}
For any district $S$ we denote $c_S$ as their center vertex. 
Fix any vertex $r$ on $G$ as the root of the tree. For each vertex $v\in V\setminus \{r\}$, we let $\mathrm{parent}(v)$ to be the parent vertex of $v$.
For brevity we also define $\mathrm{parent}(r)=\perp$.
Then, each district $S$ can be of two types: either $\mathrm{parent}(c_S)\notin S$ (denoted as \textbf{type I}) or $\mathrm{parent}(c_S)\in S$ (denoted as \textbf{type II}).
For each district $S$ we also define $\mathrm{highest}(S)$ to be the vertex in $S$ that is closest to the root on $G$.

Since the pairs of overlapping districts are unordered, it suffices to consider for each district $A$, \emph{only} the overlapping districts $B$ such that either (1) $c_A$ is NOT an ancestor of $c_B$, or (2) $c_A=c_B$ and $\mathrm{highest}(A) \neq \mathrm{parent}(\mathrm{highest}(B))$ --- this condition excludes the situation where $A$ is type II but $B$ is type I.

Let $\mathcal{S}_A$ be the collection of all overlapping districts $B$ that satisfies the above criteria.
We can now express the sum of products as:
\begin{equation}\label{eq:16}
\sum_{A\cap B\neq\emptyset} x_Ax_B \le 
    \sum_{A\in\mathcal{S}} x_A\cdot \left(
    \sum_{B\in \mathcal{S}_A} x_B
    \right)
\end{equation}
We note that some overlapping pairs $(A, B)$ may be double counted in the above notation so we may have inequality. For example, when both $A$ and $B$ are type II and $c_A, c_B$ are siblings. Another example is when $c_A=c_B$ and $A, B$ are both type I.
Now, for each $A\in \mathcal{S}$. Depending on whether $A$ is type I or type II, we have two cases:

\begin{itemize}[itemsep=0pt]
\item \textbf{$A$ is type I.} In this case, by the definition of $\mathcal{S}_A$, we know that \emph{every} district in $\mathcal{S}_A$ contains $c_A$. Therefore, by the primal constraints we have $\sum_{B\in \mathcal{S}_A}x_B \le 1$.
\item \textbf{$A$ is type II.} In this case, by the definition of $\mathcal{S}_A$, we know that all districts in $\mathcal{S}_A$ contains $\mathrm{parent}(c_A)$: whenever $c_B=c_A$, the criteria (2) guarantees that $B$ is of type II too.
Hence, using the fact that $G$ is a tree, $\mathrm{parent}(c_A)=\mathrm{parent}(c_B)\in B$.
Thus, we also have $\sum_{B\in\mathcal{S}_A}x_B\le 1$ as well.
\end{itemize}
\Cref{lem:lp-trees} immediately follows \Cref{eq:16} and the above case analysis.
\end{proof}

By setting $\tau=2$ in \Cref{lem:rounding-tau}, we conclude that the algorithm produces $(4+\epsilon)$-approximate solutions on trees.

\subsection{General Graphs}\label{sec:appendix-general-graphs}

\GeneralGraphRoundingGap*

\begin{proof}
We split the set $\mathcal{S}$ into two parts $\mathcal{S}_{\mathrm{large}} := \{S\in\mathcal{S}\ |\ |S|\ge \sqrt{n}\}$ and $\mathcal{S}_{\mathrm{small}} := \{S\in\mathcal{S}\ |\ |S|<\sqrt{n}\}$ according to the cardinality of each $c$-balanced district.
We first claim that 
\newcommand{\SLarge}{\mathcal{S}_{\mathrm{large}}}
\begin{equation}\label{eq:9}
    \sum_{A, B\in\mathcal{S}_{\mathrm{large}}:\ A\cap B\neq\emptyset} x_Ax_B \le \sqrt{n}\cdot \sum_{S\in \mathcal{S}_{\mathrm{large}}} x_S.
\end{equation}
As $\sum_{A, B\in \mathcal{S}_{\mathrm{large}}} x_Ax_B\le (\sum_{S\in \mathcal{S}_{\mathrm{large}}} x_S)^2$, it suffices to show that $\sum_{S\in\mathcal{S}_{\mathrm{large}}} x_S\le \sqrt{n}$.
Let $c_S$ be the center vertex of the district $S$.
By a double counting method, we have:
\begin{align*}
\sum_{S\in \SLarge} x_{S}\cdot \sqrt{n} &\le \sum_{S\in \SLarge} x_S\cdot |S| 
\le \sum_{S\in\SLarge} \sum_{v\in S} x_S
\leq \sum_{v\in V} \sum_{S:\ S\ni v} x_S \\
&\le \sum_{v\in V}  1 = n. \tag{by primal constraints of LP}
\end{align*}
Therefore \Cref{eq:9} is true and the equation captures the sums between all pair of large-size districts.
For every remaining pair of districts that are overlapping, one of the districts must be in $\mathcal{S}_{\mathrm{small}}$. For each district that overlaps with $A$, we charge it to one of the common vertices, arbitrarily chosen.
Let's say $A\in\mathcal{S}_{\mathrm{small}}$. Then, by partitioning all districts that overlaps with $A$ on each vertex $v\in A$, we have
\begin{equation}\label{eq:10}
\sum_{B\in\mathcal{S}: \ A\cap B\neq\emptyset} x_B \le \sum_{v\in A}\sum_{B\in\mathcal{S}: \ v\in A\cap B} x_B \le
\sum_{v\in A} 1 \le |A|.
\end{equation}
Finally, by combining both \Cref{eq:9} and \Cref{eq:10}, we have
\begin{align*}
\sum_{A, B\in \mathcal{S}:\ A\cap B\neq \emptyset} x_Ax_B 
&\le \sum_{A, B\in \SLarge:\ A\cap B\neq \emptyset} x_Ax_B 
+ \sum_{A\in\mathcal{S}_{\mathrm{small}}, B\in\mathcal{S}:\ A\cap B\neq \emptyset} x_Ax_B\\
&\le \sqrt{n}\cdot \sum_{S\in\SLarge} x_S + \sum_{A\in\mathcal{S}_{\mathrm{small}}, B\in\mathcal{S}:\ A\cap B\neq \emptyset} x_Ax_B \tag{by \Cref{eq:9}}\\
&\le \sqrt{n}\cdot \sum_{S\in\SLarge} x_S + \sum_{A\in\mathcal{S}_{\mathrm{small}}} |A|\cdot x_A
\tag{by \Cref{eq:10}}\\
&\le \sqrt{n}\cdot \sum_{S\in\SLarge} x_S + \sum_{A\in\mathcal{S}_{\mathrm{small}}} \sqrt{n}\cdot x_A\\
&\le \sqrt{n}\cdot \sum_{S\in\mathcal{S}} x_S 
\end{align*}
as desired.
\end{proof}



\subsection{Large Rounding Gap Examples}\label{sec:large-rounding-gap}

We first show that it is possible for the analysis reporting a rounding gap greater than $1$, even if the graph is planar.
We remark that, however, these examples do not imply hardness of approximation results, as it is still possible for the rounding algorithm (e.g., by testing with a smaller $\tau$ value) returning a close-to-optimal solution.

\subsubsection{Constant $>1$ Rounding Gap For Planar Graphs}

We examine the graphs which correspond to the following two grid patterns: the (square) grid and the triangular grid.

\begin{figure}[ht]
    \centering
    \includegraphics[width=0.3\textwidth]{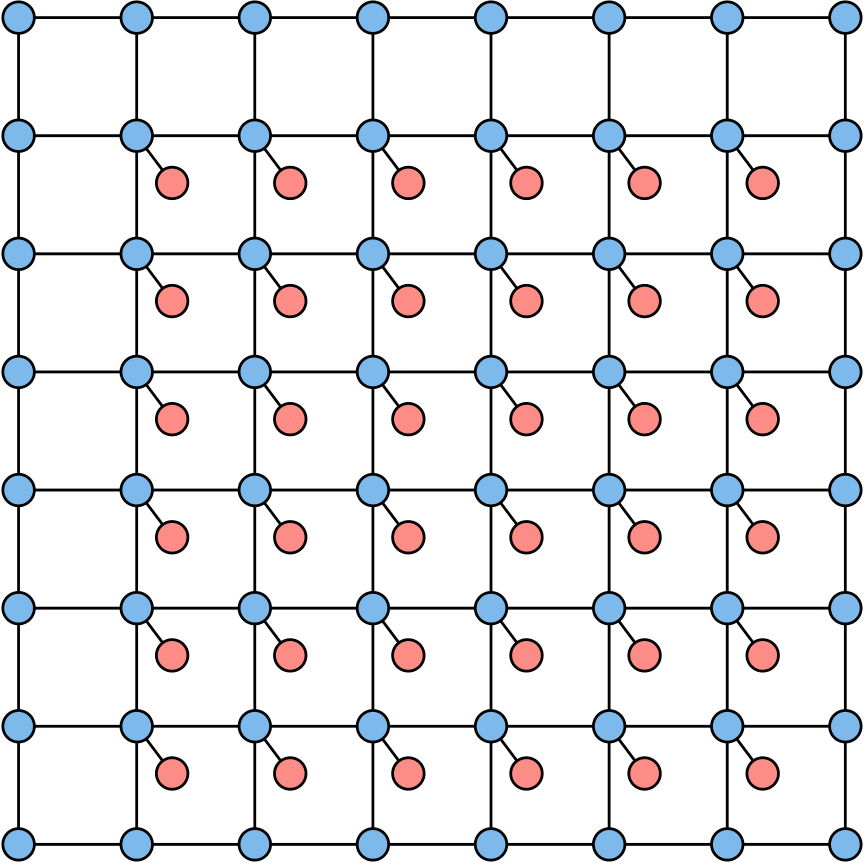}
    \hspace*{2cm}
    \includegraphics[width=0.3\textwidth]{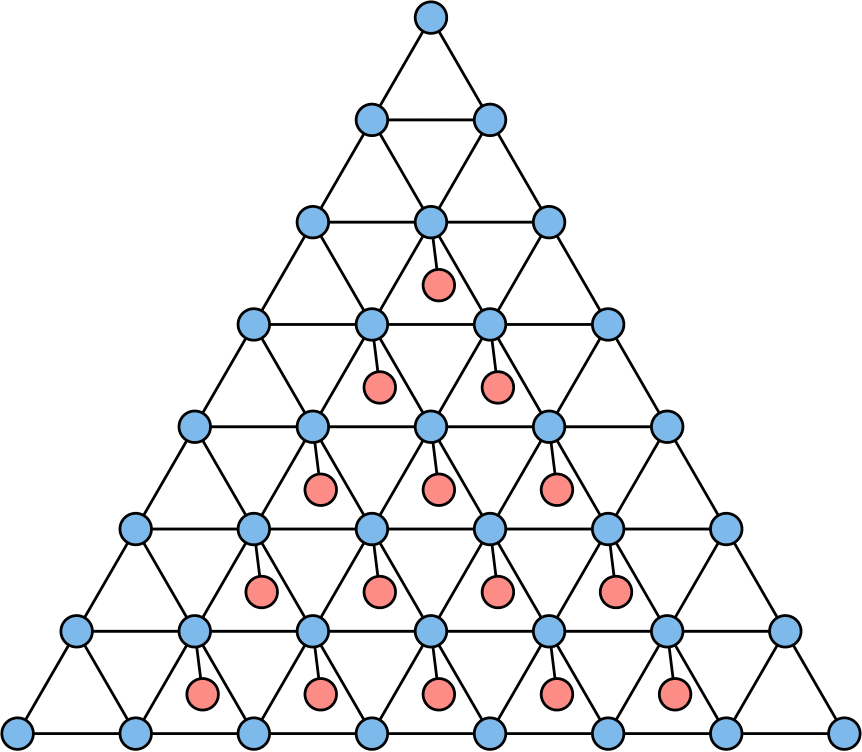}
    \caption{Construction of two types of the grid. The boundary of grid contains $O(\sqrt{n})$ vertices.
    Each internal (non-boundary) vertex is attached with an additional vertex for fulfilling the balanced-ness constraint. The weights are set as follows. (a) For the square grid construction, each grid vertex has $p_1(v)=1$ and each attached vertex has $p_2(v)=(c-1)\times 5$. (b) For the triangular grid construction, each grid vertex has $p_1(v)=1$ and each attached vertex has $p_2(v)=(c-1)\times 7$.}
    \label{fig:grids}
\end{figure}

\paragraph{Square Grid.} Consider a $\sqrt{n}\times \sqrt{n}$ square grid graph. Each grid vertex has weight $p_1(v)=1$ and $p_2(v)=0$. 
Each internal vertex is attached with an additional vertex (called a dangling vertex), with weight $p_1(v)=0$ and $p_2(v)=(c-1)\times 5$. See the left figure in \Cref{fig:grids}.
This assignment ensures that any $c$-balanced \emph{star} district $S$ has a center at an internal vertex and it must take all 5 neighbors (4 neighbors on the grid and one dangling vertex).
An integral solution can be done by greedily tiling ``cross shaped'' tiles on the grid. This achieves a total weight $(1-o(1))cn$.
On the other hand, we can set dual variables to be $y_v=c$ for all $v$ on the grid, and $y_v=0$ for all vertex $v$ that is dangling. This ensures a feasible dual solution with objective value $cn$.
Thus, the integrality gap of the formulated LP is at most $1+o(1)$.

Now, we examine the rounding algorithm with a possible output from LP.
A close-to-optimal fractional solution can be achieved by setting $x_S=1/5$ for all district $S$. The objective value is then $(\sqrt{n}-2)^2 \cdot \frac{1}{5} \cdot (5(c-1)+5) = (1-o(1))cn$.
The sum of products between overlapping districts is at least
\begin{align*}
\frac{1}{2}\cdot (\sqrt{n}-6)^2\cdot 12 \cdot \frac{1}{5^2} = (1-o(1))\cdot \frac{6}{25}\cdot n\ .
\end{align*}
Since the sum of all $x_S$ is at most $n/5 = \frac{5}{25}\cdot n$, we know that $\tau\ge 6/5$ must be set in order to apply \Cref{lem:rounding-tau}.

\paragraph{Triangular Grid.} Consider a triangular grid where each side has $O(\sqrt{n})$ vertices.
Similar to the square grid construction, we assign each grid vertex a weight $p_1(v)=1$ and $p_2(v)=0$.
For each internal vertex, there is an additional dangling vertex attached with weight $p_1(v)=0$ and $p_2(v)=(c-1)\times 7$.
See the right figure in \Cref{fig:grids}.
A similar argument to the square grid analysis shows that the integrality gap is also at most $1+o(1)$.
Now, a fractional solution can be achieved by setting $x_S=1/7$ for all $c$-balanced star district $S$.
The sum of products between overlapping districts is at least 
\[
\frac{1}{2}\cdot (1-o(1))n\cdot 18\cdot \frac{1}{7^2} = (1-o(1))\cdot \frac{9}{49}\cdot n\ .
\]
Since the sum of all $x_S$ is at most $n/7=\frac{7}{49}\cdot n$, we know that $\tau\ge 9/7$ must be set in order to apply \Cref{lem:rounding-tau}. We notice that $9/7 > 6/5$, which leads to a slightly larger gap comparing with the square grid case.

\subsubsection{Large Rounding Gap on General Graphs}
In \Cref{sec:star-districting-general-graphs} we have shown an example reduced from hypergraph matching with both the integrality gap and the rounding gap are $\Theta(\sqrt{n})$. 
Here we show an instance with $O(1)$ integrality gap yet still $\Omega(\sqrt{n})$ rounding gap.


\paragraph{The Construction.}
In this construction we assume $\sqrt{n}$ to be an integer.
We construct a bipartite graph $G=((A\cup B)\cup R, E)$ with $3n$ vertices ($|A|=|B|=|R|=n$) as follows.
We regard vertices in $R$ as a $\sqrt{n}\times \sqrt{n}$ grid.
Each vertex in $R$ can be labelled as $R_{i, j}$ where $1\le i, j\le \sqrt{n}$.
Both $A$ and $B$ are partitioned into $\sqrt{n}$ sets $\{A_i\}$ and $\{B_j\}$, each containing $\sqrt{n}$ vertices.
For each $i$, each vertex $v\in A_i$ has incident edges to all vertices $R_{i, 1}, R_{i, 2}, \ldots, R_{i, n}$.
For each $j$, each vertex $v\in B_j$ has incident edges to all vertices
$R_{1, j}, R_{2, j}, \ldots, R_{n, j}$.
The degree of each vertex in $A$ and $B$ is $\sqrt{n}$ and the degree of each vertex in $R$ is then $2\sqrt{n}$.

Now, we assign the weights. Each vertex $v$ in $A\cup B$ has weight $p_1(v)=(c-1)\sqrt{n}$ and $p_2(v)=0$.
Each vertex $v$ in $R$ has weight $p_1(v)=0$ and $p_2(v)=1$. With this weight assignment, the only $c$-balanced star districts that can be formed are the districts that are centered at any vertex $v\in A\cup B$ and containing all neighbors of $v$.
There are exactly $2n$ districts. 

\paragraph{A Fractional Solution.}
For each $v\in A\cup B$ we denote $S_v$ to be the (unique) district centered at $v$. A fraction solution for the formulated LP can be simply setting $x_{S_v}=1/(2\sqrt{n})$. This is definitely feasible since the degree of vertices in $R$ is $2\sqrt{n}$.
This solution is also optimal since the weight contribution of $p_2$ to the objective function is maximized.

\paragraph{Integrality Gap.} It is also easy to see that the integrality gap is $1$: for each set $A_i$ we pick an arbitrary vertex $v\in A_i$ and add $S_v$ to the final districting. The final districting covers all vertices in $B$ so it is an optimal solution as well.

\paragraph{Rounding Gap.} We now show that this example gives an $\Omega(\sqrt{n})$ rounding gap. Indeed, for each vertex $v\in A$, the district $S_v$ intersects with \emph{all} $S_u$ where $u\in B$: suppose $v\in A_i$ and $u\in B_j$, then they share a common vertex  $R_{i, j}\in S_v\cap S_u$. Therefore, we have
\[
\sum_{S,T:\ S\cap T\neq\emptyset} x_Sx_T \ge |A|\cdot |B|\cdot \left(\frac{1}{2\sqrt{n}}\right)^2 = \frac{n}{4}.
\]
On the other hand, the sum of all primal variables is
\[
\sum_{S} x_S = 2n\cdot \left(\frac{1}{2\sqrt{n}}\right) = \sqrt{n}.
\]
Therefore, the rounding gap is $\sqrt{n}/4=\Omega(\sqrt{n})$.

\subsubsection{Large Rounding Gap for Deterministic Greedy Rounding}\label{appendix:det-rounding}

\begin{figure}[ht]
    \centering
    \includegraphics[width=0.5\textwidth]{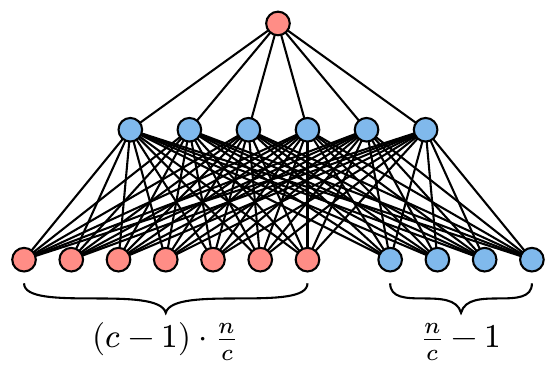}
    \caption{Construction of hard instance for greedy rounding.}
    \label{fig:greedy_round}
\end{figure}

In this section, we give an instance on which deterministic greedy rounding fails to give a good approximation, motivating the need for randomized rounding. Here in the deterministic greedy rounding, we sort the primal variables in decreasing order and choose a district as long as it does not overlap with any chosen districts. 

\paragraph{The construction:} Refer to \Cref{fig:greedy_round} for the construction. The graph can be considered as divided into three layers, with single vertex in first layer (denoted henceforth by $v$), $c-1$ vertices in second layer (call this set $X$) and $n-1$ vertices in the third layer (call this set $Y$). $v$ has $p_1$ population 1, each vertex in $X$ has $p_2$ population 1, $Y$ has $(c-1) \frac{n}{c}$ vertices with $p_1$ population 1 and $\frac{n}{c}-1$ vertices with $p_2$ population 1. Vertex $v$ is connected to every vertex in $X$ and there exists a complete bipartite graph between $X$ and $Y$. The total number of vertices in this graph is $n+c-1$.

The $c$-balanced districts are as follows: $v$ and all its neighbors in $X$ form a $c$-balanced district of weight $c$ (call this district $S_v$). Each node $u \in X$ forms a $c$-balanced district with all its neighbors in $Y$, each of which has weight $n$ (denoted by $S_u$ for $u \in X$). Note that all these districts are star shaped. Consider the fractional solution, where $x(S_u) = \frac{1}{c-1}$ for each district centered at $u \in X$ and $x(S_v) = 1 - \frac{1}{c-1}$. It's straightforward to verify that this forms a feasible solution, with objective value $n + c - \frac{c}{c-1}$, which forms a $(1+\varepsilon)$-approximate solution since the optimal solution is to pick any one district centered at a vertex in $X$. For constant $c > 3$, we have $1 - \frac{1}{c-1} > \frac{1}{c-1}$, and as a result the greedy rounding that sorts districts by $x(S)$ values for districts $S$, chooses the district $S_v$ with weight $c$ (which is a constant), and cannot choose any other districts. This gives an $\Omega(n)$-approximation.



\section{Omitted Proofs from \Cref{sec:FPTAS}}\label{sec:omitted-proofs-from-fptas}

\CompleteGraphDistrictingTheorem*

\begin{proof}[Proof of \Cref{thm:fptas-complete}]
Given an arbitrary ordering on blocks, let $L^i$ be the set of all values that can be obtained by selecting some subset of the first $i$ blocks $\{v_1, \dots, v_i\}$,
$$L^i = \left\{\exists S\subseteq [i], \sum_{j\in S}\vp(v_j)\right\}\subset \mathbb{Z}_{\ge 0}^2.$$
We use induction to show that $L^i_1$ in \cref{alg:complete} is an $(\ell_1, \frac{\epsilon i}{n})$-trimmed of $L^i$ for all $i$.    The base case $i = 0$ is trivially holds as $L^0 = \emptyset$.  Suppose $L^{i-1}_1$ is an $\left(\ell_1, \frac{\epsilon (i-1)}{n}\right)$-trimmed of $L^{i-1}$.  For any $\vq+\vp(v_i)\in L^i\setminus L^{i-1}$ with $\vq\in L^{i-1}$, by the induction hypothesis, there exists $\vq'\in L^{i-1}_1$ which  $\frac{\epsilon(i-1)}{n}$-approximates and $\ell_1$-dominates $\vq$.  Because $p_1(v_i), p_2(v_i)\ge 0$, 
\begin{align*}
    &\frac{q_1'+p_1(v_i)}{q_1+p_1(v_i)}, \frac{q_2'+p_2(v_i)}{q_2+p_2(v_i)}\in [e^{\frac{-\epsilon(i-1)}{n}}, e^{\frac{\epsilon(i-1)}{n}}]
    \text{ and }\ell_1(\vq'+\vp(v_i)) \ge \ell_1(\vq+\vp(v_i)).
\end{align*}
On the other hand, by the definition of $L^i_1$, for any $\vq'+\vp(v_i)\in L^{i-1}_1+\vp(v_i)$ there exists $\vq''\in L^i_1$ so that 
\begin{align*}
    &\frac{q_1''}{q_1'+p_1(v_i)}, \frac{q_2''}{q_2'+p_2(v_i)}\in [e^{\frac{-\epsilon}{n}}, e^{\frac{\epsilon}{n}}]
    \text{ and } \ell_1(\vq'')\ge \ell_1(\vq'+\vp(v_i))
\end{align*}
Combining these two proves that $\vq''$ $\frac{\epsilon i}{n}$-approximates and $\ell_1$-dominates $\vq+\vp(v_i)$.  The identical argument holds for all $\vq\in L^{i-1}\subseteq L^i$.  Thus, we show $L^i_1$ is an $(\ell_1, \frac{\epsilon i}{n})$-trimmed of $L^i$ for all $i$.  Similar argument applies to $L^i_2$.

Let $\vq^*$ be the optimal $c$-balanced value in $L^n$.  Suppose $q_2^*\ge q_1^*$.  Since $L^n_1$ is $(\ell_1, \epsilon)$-trimmed to $L^n$, there exists $\vq'\in L^n_1$ that $\epsilon$-approximates and $\ell_1$-dominates $\vq^*$.  Because $\vq'$ $\epsilon$-approximates $\vq^*$, the approximation guarantee holds, $q'_1+q_2' \ge e^{-\epsilon} (q_1^*+q_2^*)$.  Now we show $\vq'$ is also $c$-balanced. 
 Because $\vq^*$ is $c$-balanced and $\vq'$ $\ell_1$-dominates $\vq^*$, 
 $$0\le (c-1)q_1^*- q_2^* = \ell_1(\vq^*)\le \ell_1(\vq').$$
Moreover, because $q_2^*\ge q_1^*$ and $\vq'$ $\epsilon$-approximates $\vq^*$, we have
    $$(c-1)q_2'\ge  (c-1)e^{-\epsilon}q_2^*\ge (c-1)e^{-\epsilon}q_1^*\ge (c-1)e^{-2\epsilon}q_1'\ge q_1'$$
    where the last inequality holds because $\frac{1}{2}\ln(c-1)\ge \epsilon$.  Combining these two, we have $\ell_1(\vq')$ and $\ell_2(\vq')\ge 0$ completing the proof.  Similarly, if $q_2^*\ge q^*_1$, there exists an $\epsilon$-approximation and $c$-balanced solution in $L^n_2$.  

    The running time of $i$-th iteration is $O(|L_1^i|^2+|L_2^i|^2)$ which can be bounded as the following.  Consider a geometric grid with vertices in  $\{(e^{\frac{j}{n}\epsilon}, e^{\frac{k}{n}\epsilon}): j, k = 0,\dots, \lceil \frac{n}{\epsilon}\ln w(V)\rceil\}$.  Because $L_1^{i}\subseteq [w(V)]^2$ and no two points in can be in a same rectangle after trimming, the size of $L_1^i$ is bounded by the size of grid $O(\frac{n^2}{\epsilon^2}(\ln w(V))^2)$. Therefore, the running time of \cref{alg:complete} is $O(\frac{n^5}{\epsilon^4}(\ln w(V))^4)$.  The additional $n$ in the theorem statement is to reconstruct the set.
\end{proof}

\TreeDistrictingTheorem*

Analogous to \cref{sec:complete}, with slight abuse of notation, let 
$\ell_1(\vs) = (c-1)s_1-s_2$ and $\ell_2(\vs) = (c-1)s_2-s_1$ for all  $\vs\in \mathbb{R}^3$.  Given $j = 1,2$, $\epsilon\ge 0$, and $\vs, \vs'$, $\vs$ is \emph{$\ell_j$-dominated} by $\vs'$ if $\ell_j(\vs)\le \ell_j(\vs')$, and $\vs$ is an \emph{$\epsilon$-approximate} of $\vs'$ if 
$s_1/s_1', s_2/s_2', s_3/s_3'\in [e^{-\epsilon}, e^{\epsilon}]$ with $0/0 := 1$. 
Finally, given two sets $L, L'\subset \mathbb{R}^3$, $L'$ is a $(\ell_j,\epsilon)$-trimmed of $L$ if $L'\subseteq L$ and for all $\vs\in L$ there exists $\vs'\in L'$ which is an $\epsilon$-approximate and $\ell_j$-dominates $\vs$, and $L+L' = \{\vs+\vs': \vs\in L\text{ and } \vs'\in L'\}$. 

The following lemma shows that the parameter $\epsilon$ decay smoothly under composition and addition.
\begin{lemma}\label{lem:comp}
    Given $\epsilon, \epsilon_1, \epsilon_2\ge 0$ and $j = 1,2$, if $L_1'$ is an $(\ell_j, \epsilon_1)$-trimmed of $L_1$ and $L_2'$ is an $(\ell_j, \epsilon_2)$-trimmed of $L_2$, $L'' = \Ftrim{$L_1'+L_2', \ell_j, \epsilon$}$ is a $(\ell_j, \epsilon+\max (\epsilon_1, \epsilon_2))$-trimmed of $L_1+L_2$.   
    Moreover, if $L_1+L_2\subseteq \{0,1,\dots, W\}^3$ for some $W\in \mathbb{Z}_{> 0}$, $|L''| \le \left(1+\frac{\ln W}{\epsilon}\right)^3$.
\end{lemma}
\begin{proof}
    For the first part, $L_1'+L_2'\subseteq L_1+L_2$ is trivial.  For any $\vs_1 = (s_{1,1}, s_{1,2}, s_{1,3})\in L_1$ and $\vs_2 = (s_{2,1}, s_{2,2}, s_{2,3})\in L_2$, there exists $\vs_1'\in L_1'$ and $\vs_2'\in L_2'$ so that $\ell_j(\vs_1')\ge \ell_j(\vs_1), \ell_j(\vs_2')\ge \ell_j(\vs_2)$, 
    \begin{align*}
        &\frac{s_{1,1}}{s_{1,1}'}, \frac{s_{1,2}}{s_{1,2}'}, \frac{s_{1,3}}{s_{1,3}'}\in [e^{-\epsilon_1}, e^{\epsilon_1}]\text{, and }\frac{s_{2,1}}{s_{2,1}'}, \frac{s_{2,2}}{s_{2,2}'}, \frac{s_{2,3}}{s_{2,3}'}\in [e^{-\epsilon_2}, e^{\epsilon_2}].
    \end{align*}
    Hence, $\vs_1'+\vs_2'\in L_1'+L_2'$ and $\vs_1+\vs_2\in L_1+L_2$ satisfy $\ell_1(\vs_1'+\vs_2')\ge \ell_1(\vs_1+\vs_2)$, 
    \begin{align*}
        &\frac{s_{1,1}+s_{2,1}}{s_{1,1}'+s_{2,1}'}, \frac{s_{1,2}+s_{2,2}}{s_{1,2}'+s_{2,2}'}, \frac{s_{1,3}+s_{2,3}}{s_{1,3}'+s_{2,3}'}\in [e^{-\max \epsilon_1, \epsilon_2} e^{\max \epsilon_1, \epsilon_2}].
    \end{align*}
    Because $L''$ is $(\ell_1, \epsilon)$-trimmed of $L_1'+L_2'$, $L''\subset L_1'+L_2'\subset L_1+L_2$, and $L''$ is $(\ell_1, \epsilon+\max(\epsilon_1, \epsilon_2))$-trimmed of $L_1+L_2$ using a similar argument.

    For the second part, because $L''\subseteq L_1+L_2$, the value of $L''$ is in $\{0,1,\dots, W\}^3$.  On the other hand, consider a geometric grid with vertices in $\{\vs\in [0,V]^3: s_1, s_2, s_3 = 0, 1, e^{\epsilon}, e^{2\epsilon},\dots,\}$.  As no two points in $L''$ can be in the same rectangle, the size of $L''$ is bounded by the size of the grid, $O((\epsilon^{-1}\ln W)^3)$.  
\end{proof}

\begin{algorithm}
\caption{FPTAS on tree graphs}\label{alg:tree}
\KwData{$\epsilon>0$, $c> 2$, a rooted tree $(V,E)$, and population functions $\vp = (p_1, p_2)$ }
\KwResult{The size of optimal $c$-balanced districts}
\SetKwFunction{Ftrim}{Trim}
  \SetKwProg{Fn}{Function}{:}{}
  \Fn{\Ftrim{$L$, $\ell$, $\varepsilon$}}{
        Sort $L = \{\vs_1, \dots, \vs_m\}$ so that $\ell(\vs_1)\ge \ell(\vs_2)\ge\dots\ge  \ell(\vs_m)$\;
        Set $L_{out} = \emptyset$ and $L_{rem} = \emptyset$\;
        \For{$i = 1,\dots, m$}
        {
            \If{$\vs_i\notin 
            L_{rem}$}{
            $L_{out}\gets L_{out}\cup \{\vs_i\}$\;
            \For{$j = i+1,\dots,m$}
            {
                \If{$\vs_i$ $\varepsilon$-approximates $\vs_j$}{
                    $L_{rem}\gets L_{rem}\cup\{\vs_j\}$
                }
            }
            }
        }
        \KwRet $L_{out}$\;
  }
\SetKwFunction{Fgrow}{grow}
  \SetKwProg{Fn}{Function}{:}{}
  \Fn{\Fgrow{$v$, $\varepsilon$}}{
        Set $L_1^v, L_2^v\gets \emptyset$\;
        Set $I_1, I_2 \gets \{(p_1(v), p_2(v), 0)\}$ \tcp{Only include itself}
        \For{Each $v$'s child $u$}
        {
            $(L_1^u, L_2^u)\gets \Fgrow{$u, \varepsilon$}$\;
            $I_1\gets \Ftrim{$I_1+L_1^u, \ell_1, \varepsilon$}$\;
            $I_2\gets \Ftrim{$I_2+L_2^u, \ell_2, \varepsilon$}$\;
        }
        $C\gets \{(0,0,s_1+s_2+s_3): \exists \vs\in I_1\cup I_2, \ell_1(\vs), \ell_2(\vs)\ge 0\}$\;
        $A\gets\{(0,0,s_3):\exists \vs\in I_1\cup I_2 \}$\;
        \uIf{$v\neq root$}{
        $L_1^v \gets \Ftrim{$I_1\cup C\cup A, \ell_1, \varepsilon$}$\;
        $L_2^v \gets \Ftrim{$I_2\cup C\cup A, \ell_2, \varepsilon$}$\;
    }
    \Else{
    $L_1^v \gets \Ftrim{$C\cup A, \ell_1, \varepsilon$}$\;
    $L_2^v \gets \Ftrim{$C\cup A, \ell_2, \varepsilon$}$\;
    }
        \KwRet $(L_1^v, L_2^v)$\;
  }
$L_1^{root}, L_2^{root}\gets \Fgrow{$root, \epsilon/n$}$\;
\KwRet $\max\{s_3: \vs\in L_1^{root}, L_2^{root}\}$
\end{algorithm}

\begin{proof}[Proof of \cref{thm:tree}]
    Note that the size of an optimal solution $\mathcal{T}^*$ is $s_3^{root}(\mathcal{T}^*)$, and $(0,0,s_3^{root}(\mathcal{T}^*))\in L^{root}$.  If $L_{1}^{root}$ and $L_{2}^{root}$ are $(\ell_1, \epsilon)$-trimmed and $(\ell_2,\epsilon)$-trimmed of $L^{root}$ respectively, there exists $\vs\in L_1^{root}\subseteq L^{root}$ such that $s_1 = s_2 = 0$ and $s_3\ge e^{-\epsilon}s_3^{root}(\mathcal{T}^*)$ which yields an approximation ratio of $e^\epsilon$ and completes the proof.
    To this end, let $size(v)$ denote the number of block in the subtree of $v$ including $v$ and $\epsilon_v := size(v)\frac{\epsilon}{n}\le \epsilon$ with $\epsilon_v' = (size(v)-1)\frac{\epsilon}{n}$.  
    We will use induction.  As we cannot decide whether the current incomplete district is $c$-balanced or not in \Fgrow, our induction will be on slightly extended sets instead of $L^v$, defined as following
 \begin{equation}
     \bar{L}^v = \left\{\vs^v(\mathcal{T}):  \text{\parbox{8cm}{$\mathcal{T} = (T_1, \dots, T_m)$ and for all $i$, $T_i$ is $c$-balanced
 if it is fully contained in the subtree of $v$.}}\right\}\label{eq:tree_stamp_ex}
 \end{equation}
 where $L^v\subseteq \bar{L}^v$ and $L^{root} = \bar{L}^{root}$ because $\bar{L}^v$ allows $\mathcal{T}$ to contain unbalanced districts outside of the subtree of $v$.  Additionally, we partition $\bar{L}^v$ into three sets: 
 \begin{align*}
     I^v=& \left\{\vs^v(\mathcal{T})\in \bar{L}^v: \text{{$v$ is incomplete for some $T\in \mathcal{T}$}}.\right\}\\
     C^v=& \left\{\vs^v(\mathcal{T})\in \bar{L}^v: \text{{$v$ is consolidating}}.\right\}\\
     A^v=& \left\{\vs^v(\mathcal{T})\in \bar{L}^v: \text{{$v$ is absent}}.\right\}
 \end{align*}
We use induction to show that {$L_1^v$ is a $(\ell_1, \epsilon_v)$-trimmed of $\bar{L}^v$ and $L_2^v$ is a $(\ell_1, \epsilon_v)$-trimmed of $\bar{L}^v$ for all $v$.}

    When $v$ is a leaf, the statement holds.  Now suppose $v$ is not the root and for all $v$'s child $u$, $L_1^u$ is a $(\ell_1, \epsilon_u)$-trimmed of $\bar{L}^u$ and $L_2^u$ is a $(\ell_2, \epsilon_u)$-trimmed of $\bar{L}^u$.  
    
    First, because $I^v$ can be written as 
    $$I^v=(p_1(v), p_2(v),0)+\sum_{\text{child } u}\bar{L}^u,$$
    and $I_1$ is $(p_1(v), p_2(v),0)+\sum_{\text{child } u} L_1^u$ applied \Ftrim $deg(v)-1$ times.
    Thus $I_1$ (and $I_2$) is an $(\ell_1, \epsilon_v')$-trimmed (and $(\ell_2, \epsilon_v')$-trimmed) of $I^v$, due to \cref{lem:comp} and $\max_u \epsilon_u+\frac{\epsilon}{n}(deg(v)-1)\le \epsilon_v'$. 

    Second, we show that $A$ is an $(\ell_1, \epsilon_v')$ and  $(\ell_2, \epsilon_v')$-trimmed of $A^v$.  For any districting $\mathcal{T}$ satisfying the condition in \cref{eq:tree_stamp_ex}, if $v$ is absent in $\mathcal{T}$ so that $\vs^v(\mathcal{T})\in A^v$, we can create a new districting $\mathcal{T}' = \mathcal{T}\cup \{v\}$ so that $v$ is incomplete in $\mathcal{T}'$ and thus $\vs^v(\mathcal{T}')\in I^v$.  Hence
    $$s_1^v(\mathcal{T}) = s_2^v(\mathcal{T}) = 0\text{ and } s_3^v(\mathcal{T}) = s_3^v(\mathcal{T}').$$
    Because $I_1$ is an $(\ell_1, \epsilon_v')$-trimmed of $I^v$, $I_1\subseteq I^v$ and $I_1$ has $\vz' = (z_1', z_2', z_3')$ so that $\epsilon_v'$-approximates $\vs^v(\mathcal{T}')$ so that $$e^{-\epsilon_v'}\le \frac{z_3'}{s_3^v(\mathcal{T}')}\le e^{\epsilon_v'}.$$  
    By the definition of $A$ in \cref{alg:tree}, there exists $\vz\in A$ so that $z_1 = z_2 = 0$ and $z_3 = z_3'$.  Combining these we have 
    $$s_1^v(\mathcal{T}) = z_1 = s_2^v(\mathcal{T}) = z_2 = 0\text{ and }e^{-\epsilon_v'}\le \frac{z_3}{s_3^v(\mathcal{T})}\le e^{\epsilon_v'}.$$
    Because $I_1, I_2\subseteq I^v$, $A\subseteq A^v$.

    Third, we show that $C$ is an $(\ell_1, \epsilon_v')$ and  $(\ell_2, \epsilon_v')$-trimmed of $C^v$.  If $v$ is consolidating for $T\in \mathcal{T}$, we creates $\mathcal{T}'$ so that $v$ is incomplete by extending $T$ outside the subtree of $v$ and thus $\vs^v(\mathcal{T}')\in I^v$.  Note that 
    $s^v_1(\mathcal{T}) = s^v_2(\mathcal{T}) = 0$, $s_1^v(\mathcal{T}') = p_1(T)$,  $s_2^v(\mathcal{T}') = p_2(T)$ and 
    $$s_3^v(\mathcal{T}) = s_1^v(\mathcal{T}')+s_2^v(\mathcal{T}')+s_3^v(\mathcal{T}').$$
    If $p_1(T)\le p_2(T)$, because $I_1$ is an $(\ell_1, \epsilon_v')$-trimmed of $I^v$, there exists $\vz'\in I_1$ which $\epsilon_v'$-approximates and $\ell_1$-dominates $\vs(\mathcal{T}')$.  Thus, 
    $z_1'+z_2'+z_3'\ge e^{\epsilon_v'}\left(s_1(\mathcal{T}')+s_2(\mathcal{T}')+s_3(\mathcal{T}')\right)$.
    Now we show that $(z_1',z_2')$ is $c$-balanced, so $(0,0,z_1'+z_2'+z_3')\in C$.  Because $\vz'$ $\ell_1$-dominates $\vs(\mathcal{T}')$ $\ell_1(\vz')\ge \ell_1(\vs(\mathcal{T}'))\ge 0$.  Because $p_1(T)\le p_2(T)$, we have 
    $$(c-1)z_2'\ge (c-1)e^{-\epsilon_v'} s_2(T)\ge (c-1)e^{-\epsilon_v'} s_1(T)\ge (c-1)e^{-2\epsilon_v'}z_1'\ge z_1'$$ because $(c-1)\ge e^{2\epsilon} \ge e^{2\epsilon_v'}$.  
    Therefore, $(0,0,z_1'+z_2'+z_3')\in C$ $\epsilon_v'$-approximates, $\ell_1$ and $\ell_2$ dominates $\vs^v(\mathcal{T})\in C^v$.      
    Similarly if $p_1(T)\ge p_2(T)$, there exists $\vz''\in I_2$ so that $(0,0, z_1''+z_2''+z_3'')\in C$ $\epsilon_v'$-approximates, $\ell_1$ and $\ell_2$-dominates $\vs(\mathcal{T})$.

    Combining these three cases, $L_1^v = \Ftrim{$I_1\cup A\cup C, \ell_1, \epsilon/n$}$ is an $(\ell_1, \epsilon_v)$-trimmed of $\bar{L}^v$ because $\epsilon_v'+\epsilon/n = \epsilon_v$.  The argument for $v = root$ is similar. 

    For time complexity, by \cref{lem:comp}, the size of $I_1, I_2, L_1^v, L_2^v$ and thus $A$ and $C$ are bounded by $O\left((\frac{n\ln W}{\epsilon})^3\right)$ with $W = w(V)$.  Each trimming process $\Ftrim$ takes quadratic time in the size, and there are $n$ recursive calls of $\Fgrow$, so the running time is bounded by $O(n(\frac{n\ln W}{\epsilon})^6)$. The additional $n$ in the theorem statement is to reconstruct the districting.  
\end{proof}

\end{document}